\documentclass[a4paper,12pt]{article}
\usepackage[T1]{fontenc}
\usepackage[latin1]{inputenc}
\usepackage{amsmath}
\usepackage{amssymb}
\usepackage{amsfonts}
\usepackage{euscript}
\usepackage{fancyhdr}
\usepackage{graphicx}
\usepackage{color}
\usepackage{watermark}
\usepackage{fancybox}
\usepackage{colortbl}
\usepackage{empheq}
\usepackage[english]{babel}
\usepackage[english]{minitoc}
\usepackage{hyperref}
\usepackage[ruled,vlined,linesnumbered]{algorithm2e}

\makeindex 
\newtheorem{theorem}{Theorem}[section]

\newtheorem{definition}[theorem]{Definition}
\newtheorem{example}[theorem]{Example}
\newtheorem{lemma}[theorem]{Lemma}
\newtheorem{notation}[theorem]{Notation}
\newtheorem{proposition}[theorem]{Proposition}
\newtheorem{remark}[theorem]{Remark}
\newenvironment{proof}[1][Proof]{\noindent\textbf{#1.} }{\ \rule{0.5em}{0.5em}}

\begin{document}

\title{Solving Systems of Algebraic Equations Over Finite Commutative Rings
and Applications}
\author{Hermann Tchatchiem Kamche \thanks{Hermann Tchatchiem Kamche (hermann.tchatchiem@gmail.com) is with the Institute of mathematics, University of Neuchatel, Switzerland. He is supported by the Swiss Government Excellence Scholarship (ESKAS No. 2022.0689)}
~ and Herv\'e Tal\'e Kalachi 
\thanks{Herv\'e Tal\'e Kalachi (http://perso.prema-a.org/herve.tale-kalachi/) is with the National Advanced School of Engineering of Yaound\'e, University of Yaounde I, Cameroon. He is supported by the UNESCO-TWAS and the German Federal Ministry of Education and Research (BMBF) under the SG-NAPI grant number 4500454079.}
}

\maketitle

\begin{abstract}
Several problems in algebraic geometry and coding theory over finite rings
are modeled by systems of algebraic equations. Among these problems, we have the rank decoding problem, which is used in the construction of public-key cryptography. In 2004, Nechaev and Mikhailov proposed two methods for
solving systems of polynomial equations over finite chain rings. These
methods used solutions over the residual field to construct all solutions
step by step. However, for some types of algebraic equations, one simply
needs partial solutions. In this paper, we combine two existing approaches
to show how Gr\"obner bases over finite chain rings can be used to solve
systems of algebraic equations over finite commutative rings. Then, we use
skew polynomials and Pl\"ucker coordinates to show that some algebraic
approaches used to solve the rank decoding problem and the MinRank problem
over finite fields can be extended to finite principal ideal rings.
\end{abstract}

\bigskip \textbf{Keywords}: Finite commutative rings, Gr\"obner bases,
MinRank problem, Rank decoding problem, Systems of algebraic equations.

\section{Introduction}

Solving systems of algebraic equations has always been of high interest in
algorithmic algebra. Indeed, many algebraic problems have their solution
sets contained in those of systems of algebraic equations. A tangible
example is the rank decoding problem \cite{Gaborit2016complexity}, which has
attracted a lot of attention this last decade in view of its application in
cryptography. This problem is generally defined over finite fields and
therefore, leads to the problem of solving systems of algebraic equations
over finite fields when modeled appropriately. But it should be remembered
that this latest problem has been studied for a long time and has a wide
variety of algorithms that can be used to solve it and also estimate the
solving complexities \cite{Gianni1989algebrric,Faugere1999new,Faugere2002new,Courtois2000efficient,Caminata2023solving}.

Most recently, the rank decoding problem has been extended to finite
principal ideal rings in \cite{Kalachi2023OnTheRank} where the authors,
after having justified the interest of studying this problem over finite
rings, show that it is at least as hard as the rank decoding problem over
finite fields, and also provide a combinatorial type algorithm for solving
this new problem. The translation of the rank decoding problem over finite
rings as a system of algebraic equations naturally induces the problem of
solving systems of algebraic equations over finite rings.

Contrary to the problem of solving systems of algebraic equations over
finite fields, the previous problem over finite rings has not experienced
much development. The most advanced and recent work is the paper of Mikhailov and Nechaev \cite{Mikhailov2004solving}, who proposed two approaches for
solving systems of polynomial equations over finite chain rings. One of
these approaches uses canonical generating systems, which are not Gr\"obner
bases in general. An algebraic modeling of the rank decoding problem over
finite chain rings that we will use is a system of algebraic equations with
some parameters, and we just need a partial solution. Note that Gr\"obner
bases over fields are generally used to solve these kinds of systems. A
natural question is therefore to know whether Gr\"obner bases can be used to
solve systems of algebraic equations over finite chain rings in general, as
in the case of finite fields.

Independently, Gr\"obner bases over finite chain rings have been much
studied and implemented in some mathematical software systems like Magma
\cite{Cannon2013handbook}, SageMath \cite{Sagemath2023}, etc. Indeed,
similar to Buchberger's algorithm over fields \cite%
{Buchberger1965algorithmus}, Norton and Salagean \cite%
{Norton2001strong_cyclic} gave an algorithm for computing Gr\"obner bases
over finite chain rings. This algorithm has been improved in \cite%
{Hashemi2013applying} by adding the product criterion and the chain
criterion. In the Magma handbook \cite{Cannon2013handbook}, it was specified
that the $F_{4}$ algorithm \cite{Faugere1999new} was extended over Euclidean
rings\footnote{%
Note that Euclidean rings in Magma also contain rings with zero divisors
like Galois rings.}, taking into account the elimination criteria given in
\cite{Moller1988construction}. Moreover, the elimination theorem, which is
the main property used to solve systems of algebraic equations, can be
extended over finite chain rings. However, the elimination theorem does not
hold in general on other types of finite rings. But we must not forget that
Low Rank Parity Check codes which are potential linear codes for rank-based
cryptography have been extended to finite commutative rings \cite%
{Renner2021low, Kamche2021low}. Thus, it also becomes necessary to tackle
the resolution of systems of algebraic equations over finite commutative
rings.

According to the structure theorem for finite commutative rings \cite%
{Mcdonald1974finite}, every finite commutative ring is isomorphic to a
product of finite commutative local rings. Thus, solving systems of
algebraic equations over finite commutative rings is reduced to finite local
rings. In \cite{Bulyovszky2017polynomial}, Bulyovszky and Horv\'ath gave a
good method for solving systems of linear equations over finite local rings.
Indeed, they transformed systems of linear equations from local rings to
Galois rings and used the Hermite normal form to solve it. In this work we
show that this transformation can be applied to systems of algebraic
equations, and we then use Gr\"obner bases to solve the resulting equation
since Galois rings are specific cases of finite chain rings.

Before one can use Gr\"obner bases over finite chain rings to solve the rank
decoding problem, it is first necessary to give an algebraic modeling. As
specified in \cite{Kalachi2023OnTheRank}, some properties of the rank of a
matrix over fields cannot be extended to rings in general due to zero
divisors. Therefore, the algebraic modeling of the rank decoding problem
given in \cite{BBCGPSTV20} using the MaxMinors cannot be directly applied to
rings. However, in \cite{Gaborit2016complexity} other algebraic modeling
using linearized polynomials has been given and some main properties of
linearized polynomials have been extended in \cite{Kamche2019rank} over
finite principal ideal rings. We will use these results to prove that the
algebraic modeling done in \cite{Gaborit2016complexity} using linearized
polynomials can be generalized over finite principal ideal rings.
Furthermore, as the rank decoding problem reduces to the MinRank problem
\cite{Faugere2008cryptanalysis}, we also study possible algebraic modelings of
the MinRank problem over finite rings.

The MinRank problem have several algebraic modeling over fields. For
example, the MaxMinors modeling \cite{Faugere2010computing}, the
Kipmis-Shamir modeling \cite{Kipnis1999cryptanalysis}, or the Support-Minors
modeling \cite{BBCGPSTV20}. Over finite chain rings, the rank of a matrix is
not generally equal to the order of the highest order non-vanishing minor.
Thus, the MaxMinors modeling cannot directly extend over rings. However, we
will use the rank decomposition and the Pl\"ucker coordinates to show that the
Kipmis-Shamir modeling and the Support-Minors modeling can be extended to
finite principal ideal rings.

The rest of the paper is organized as follows. In Section \ref{GrobnerBases}, we give some
properties of Gr\"obner bases over finite chain rings, followed by the use of Gr\"obner bases for solving systems of algebraic equations over finite chain
rings in Section \ref{SolvingFCR}.
In Section \ref{LocalRings} we show how to solve systems of algebraic equations over finite commutative local rings by decomposing them as a
direct sum of cyclic modules over Galois rings. Section \ref{MinRankproblem}
uses the fact that the row span of a matrix is contained in a free module of the same rank to prove that the Kipmis-Shamir Modeling and the Support
Minors Modeling of the MinRank problem can be extended to finite principal
ideal rings. In Section \ref{RDproblem}, skew polynomials are used to give an
algebraic modeling of the rank decoding problem over finite principal ideal
rings, and to finish, we conclude the paper and give some perspectives in Section \ref{conclusion}.

\section{Gr\"obner Bases Over Finite Chain Rings \label{GrobnerBases}}

\subsection{Finite Chain Rings}

A chain ring\textbf{\ }is a ring whose ideals are linearly ordered by
inclusion, and a local ring is a ring with exactly one maximal ideal. By \cite%
{Mcdonald1974finite}, a finite ring is a chain ring if and only if it is a
local principal ideal ring, that is to say a finite ring that admits exactly one
maximal ideal and every ideal is generated by one element. A basic example
of finite chain rings is the ring $\mathbb{Z}_{p^{k}}=\mathbb{Z}/p^{k}\mathbb{Z}$ of integers modulo a power of a prime
number $p$. Its maximal ideal is $p\mathbb{Z}_{p^{k}}$. Other
examples of finite chain rings that we will use to give a representation of
finite commutative local rings in Section \ref{LocalRings} are Galois rings.
A Galois ring of characteristic $p^{k}$ and rank $r$, denoted by $GR\left(
p^{k},r\right) $, is the ring $\mathbb{Z}_{p^{n}}\left[ X\right] /\left(
f\right) $, where $f$ $\in $ $\mathbb{Z}_{p^{k}}\left[ X\right] $ is a monic
polynomial of degree $r$, irreducible modulo $p$, and $\left( f\right) $
being the ideal of $\mathbb{Z}_{p^{k}}\left[ X\right] $ generated by $f$. Thus, $GR\left( p^{k},r\right) $ is a degree $r$
Galois extension of $\mathbb{Z}_{p^{k}}$ and is a finite chain
ring with maximal ideal generated by $p$ and residue field $%
\mathbb{F}_{p^{r}}=GR\left( p^{n},r\right) /pGR\left( p^{n},r\right) $ \cite%
{Mcdonald1974finite}.

In this section, we assume that $R$ is a finite commutative chain ring with
maximal ideal $\mathfrak{m}$ and residue field $\mathbb{F}_{q}=R/\mathfrak{m}$.
We denote by $\pi$ a generator of $\mathfrak{m}$, and $\nu $ the nilpotency
index of $\pi $, i.e., the smallest positive integer such that $\pi ^{\nu
}=0 $. An important property of finite chain rings is the structure of their
ideals. Every ideal of $R$ is of the form $\pi ^{i}R$, for $i=0,\ldots ,\nu $%
. A direct consequence is the following two decompositions of any element
from $R$. Let $\Gamma $ be a complete set of representatives of the
equivalence classes of $R$ under the congruence modulo $\pi $. As in \cite%
{Mikhailov2004solving}, we have for example $\Gamma =\left\{ a\in
R:a^{q}=a\right\} $.

\begin{proposition}
\label{DecompositionFCR}Let $c$ in $R$, then

(a) there exist a unique $l$ in $\left\{ 0,\ldots ,\nu \right\} $ and a unit $u$
in $R$ such that
\begin{equation}
c=\pi ^{l}u  \label{ValDecom}
\end{equation}

(b) the element $c$ has a unique representation in the form
\begin{equation}
c=\sum_{j=0}^{\nu -1}c_{j}\pi ^{j}  \label{AdicDecom}
\end{equation}%
\ where $c_{j}\in $ $\Gamma $, for $j=0,\ldots ,\nu -1$.
\end{proposition}

The integer $l$ in equation \eqref{ValDecom} is called the valuation of $c$ and is denoted $val\left( c\right)$, while the representation of $c$ given by equation \eqref{AdicDecom} is called the $\pi-$adic decomposition of $c$. Let $j\in \{0,\ldots ,\nu -1\}$, the map $\gamma
_{j}:R\longrightarrow \Gamma $ given by $c\longmapsto c_{j}$ can be extended
to elements $\mathbf{w}=\left( w_{1},\ldots ,w_{k}\right) $ from $R^{k}$ by $%
\gamma _{j}\left( \mathbf{w}\right) =\left( \gamma _{j}\left( w_{1}\right)
,\ldots ,\gamma _{j}\left( w_{k}\right) \right) $. Hence, $\mathbf{w}%
=\sum_{j=0}^{\nu -1}\pi ^{j}\gamma _{j}\left( \mathbf{w}\right) $. For $%
l=1,\ldots ,\nu -1$, set $\mathbf{w}^{[l]}=\sum_{j=0}^{l-1}\pi ^{j}\gamma
_{j}\left( \mathbf{w}\right) $. The valuation does not depend on the choice
of $\pi $ but the $\pi -$adic decomposition depends on the choice of $\pi $.

\begin{example}
The ring $%
\mathbb{Z}
_{8}$ is a finite chain ring where the maximal ideal is generated by $2$, with
 nilpotency index $3$. The residual field of $%
\mathbb{Z}
_{8}$ is $\mathbb{F}_{2}=%
\mathbb{Z}
_{8}/2%
\mathbb{Z}
_{8}$ and a complete set of representatives of the equivalence classes of $%
\mathbb{Z}
_{8}$ under the congruence modulo $2$ is $\Gamma =\left\{ 0,1\right\} $. We
have $6=3\times 2^{1}$ and $3$ is invertible in $%
\mathbb{Z}
_{8}$, thus $val\left( 6\right) =1$, moreover the maximal ideal is also
generated by $6$. The $2-$adic decomposition of $6$ is $6=0\times
2^{0}+1\times 2^{1}+1\times 2^{2}$ and the $6-$adic decomposition of $6$ is $%
6=0\times 6^{0}+1\times 6^{1}+0\times 6^{2}$.
\end{example}

The decomposition (\ref{ValDecom}) using the valuation will be used in the
next subsection to compute Gr\"obner bases and the $\pi -$adic
decomposition (\ref{AdicDecom}) will be used is Section \ref{SolvingFCR} to
solve algebraic equations.

\subsection{Gr\"obner Bases}

The ring of polynomials with $k$ indeterminates $x_{1},\ldots ,x_{k}$ and
coefficients in $R$ is denoted $R\left[ x_{1},\ldots ,x_{k}\right] $. A
monomial is an element of $R\left[ x_{1},\ldots ,x_{k}\right] $ of the form $%
x^{\alpha }:=x_{1}^{d_{1}}\cdots x_{k}^{d_{k}}$ $\ $where  the $d_{i}$'s are
non-negative integers and $\alpha =\left( d_{1},\ldots, d_{k}\right) $. If ``$>$''
is an admissible order on the set of monomials, then any element $f$ in $R%
\left[ x_{1},\ldots ,x_{k}\right] \backslash \left\{ 0\right\} $ can be
written uniquely as $f=\sum_{i=1}^{s}c_{i}x^{\alpha _{i}}$ where each $x^{\alpha
_{i}}$ is a monomial, $c_{i}\in R$, and $x^{\alpha _{1}}>\cdots >x^{\alpha
_{s}}$. The leading monomial, the leading coefficient, and the leading term
of $f$ are respectively defined by $lm\left( f\right) :=x^{\alpha _{1}}$, $%
lc\left( f\right) :=c_{1}$, $lt\left( f\right) :=c_{1}x^{\alpha _{1}}$. For $%
W\subset R\left[ x_{1},\ldots ,x_{k}\right] $, we denote by $lt\left(
W\right) $ the ideal generated by $\left\{ lt\left( w\right) \ |\ w\in
W\right\} $. According to \cite[Definition 3.8]{Norton2001strong}, we have
the following:

\begin{definition}
Let $I$ be an ideal in $R\left[ x_{1},\ldots ,x_{k}\right] $ and $G$ a
subset of $I$.

(a) $G$ is called a Gr\"obner basis for $I$ if $lt(G)=lt(I)$.

(b) $G$ is called a strong Gr\"obner basis for $I$ if for all $f\in I$
there exists $g\in G$ such that $\ lt(g)$ divides $lt\left( f\right) $, that
is to say $lt\left( f\right) =cx^{\alpha }lt(g)$ where $c\in R$ and $x^{\alpha }$
is a monomial.
\end{definition}

In \cite[Proposition 3.9]{Norton2001strong} a connection between Gr\"obner
bases and strong Gr\"obner bases was given over finite chain rings.

\begin{proposition}
\label{EquivalenceGrobnerBases}A subset of $R\left[ x_{1},\ldots ,x_{k}%
\right] $ is a Gr\"obner basis if and only if it is a strong Gr\"obner
basis.
\end{proposition}

Over finite chain rings, the reduction is also equivalent to the strong
reduction \cite[Proposition 3.2]{Norton2001strong_cyclic}. So, we will
simply recall the definition of strong reduction given in \cite[Definition
3.1 ]{Norton2001strong}.\

\begin{definition}
Let $f\in R\left[ x_{1},\ldots ,x_{k}\right] \backslash \left\{ 0\right\} $,
$h\in R\left[ x_{1},\ldots ,x_{k}\right] $ and $F=\left\{ f_{1},\ldots
,f_{n}\right\} \subset R\left[ x_{1},\ldots ,x_{k}\right] \backslash \left\{
0\right\} $. \

(i) We say that $f$ strongly reduces to $h$ with respect to $F$ in one step and 
denote it by $f\twoheadrightarrow _{F}h$, if there exists $i\in \left\{ 1,\ldots
,n\right\} $ such that $h=f-cx^{\alpha }f_{i}$, where $c\in R$, $\alpha \in
\mathbb{N}
^{k}$, and $lt\left( f\right) =cx^{\alpha }lt\left( f_{i}\right) $.

(ii) The reflexive and transitive closure of the relation $%
\twoheadrightarrow _{F}$ is denoted $\twoheadrightarrow _{F}^{\ast }$ and we
set $0\twoheadrightarrow _{F}^{\ast }0$ by convention.
\end{definition}

The reduction process and $S-$polynomials are used over fields to give a
characterization of Gr\"obner bases. But over a finite chain ring which is
not a field, we have zero divisors. Consequently, to give a characterization
of Gr\"obner bases, we also need $A-$polynomials that we define below, together with $S-$polynomials.

\begin{definition}
\label{SpolyApoly}
(i) Let $g_{1}\neq g_{2}\in R\left[ x_{1},\ldots ,x_{k}%
\right] \backslash \left\{ 0\right\} $. An S-polynomial of $g_{1}$ and $%
g_{2} $, denoted $Spoly(g_{1},g_{2})$, is any polynomial of the form $%
c_{1}x^{\alpha _{1}}g_{1}-c_{2}x^{\alpha _{2}}g_{2}$ where,
for $i\in \left\{ 1,2\right\} $, $x^{\alpha _{i}}$ is the least common
multiple of $lm\left( g_{1}\right) $ and $lm\left( g_{2}\right) $ divided by
$lm\left( g_{i}\right) $, that is, $x^{\alpha _{i}}=lcm\left( lm\left(
g_{1}\right) ,lm\left( g_{2}\right) \right) /lm\left( g_{i}\right) $, and
the constants $c_{1}$, $c_{2}$ are elements of $R$ such that $c_{1}lc\left(
g_{1}\right) =c_{2}lc\left( g_{2}\right) =lcm\left( lc\left( g_{1}\right)
,lc\left( g_{2}\right) \right) $.

(ii) Let $g\in R\left[ x_{1},\ldots ,x_{k}\right] \backslash \left\{
0\right\} $. An A-polynomial of $g$, denoted $Apoly(g)$, is any polynomial
of the form $ag$ where $a$ generates the ideal $Ann\left( lc\left( g\right)
\right) =\left\{ b\in R:b \cdot lc\left( g\right) =0\right\} .$
\end{definition}

\begin{remark}
Since $R$ is a finite chain ring, it is easy to find the constants $c_{1}$, $%
c_{2}$, and $a$ of Definition \ref{SpolyApoly}. Indeed:

(a) For $i\in \left\{
1,2\right\} $, set $lc\left( g_{i}\right) =b_{i}\pi ^{l_{i}}$ where $l_{i}=val\left( lc\left( g_{i}\right) \right) $. Then, a least common multiple of $lc\left( g_{1}\right) $ and $%
lc\left( g_{2}\right) $ is $lcm\left( lc\left( g_{1}\right) ,lc\left(
g_{2}\right) \right) =\pi ^{l}$ where $l=\max \left\{ l_{1},l_{2}\right\} $.
Thus, we can choose \ $c_{i}=b_{i}^{-1}\pi ^{l-l_{i}}$, for $i\in \left\{
1,2\right\} $.

(b) Since $\nu $ is the nilpotency index of $\pi $, then we can choose $%
a=\pi ^{\nu -val(lc(g))}$.
\end{remark}

By \cite[Corollary 5.13]{Norton2001strong} or \cite[Theorem 234]%
{Yengui2015constructive}, we have the following charaterization of Gr\"o%
bner bases over finite chain rings.

\begin{theorem}
\label{CharacterizationGB}Let $I$ be an ideal in $R\left[ x_{1},\ldots ,x_{k}%
\right] $, $I\neq \left\{ 0\right\} $, and $G$ a finite subset of $I$. Then $%
G$ is a Gr\"obner basis for $I$ if and only if the following two
conditions are satisfied:

(i) for all $g_{1}$, $g_{2}\in G$, with $g_{1}\neq g_{2}$, $\
Spoly(g_{1},g_{2})\twoheadrightarrow _{G}^{\ast }0$;

(ii) for all $g\in G$, $Apoly(g)\twoheadrightarrow _{G}^{\ast }0$.
\end{theorem}

Similar to Buchberger's algorithm over fields, Theorem \ref%
{CharacterizationGB} is used in \cite[Algorithme 3.9]%
{Norton2001strong_cyclic} to compute Gr\"obner bases.

\begin{example}
\label{ExGB} Consider $\ g_{1}=4x^{2}y+y^{3}+2y+4$ and $g_{2}=4xy^{2}$ in $%
\mathbb{Z}
_{8}[x,y]$. With lexicographic order $x>y$, we have:

$g_{3}=Spoly(g_{1},g_{2})=yg_{1}-xg_{2}=y^{4}+2y^{2}+4y$,

$g_{4}=Apoly(g_{1})=2g_{1}=2y^{3}+4y$,

$g_{5}=Spoly(g_{1},g_{3})=y^{3}g_{1}-4x^{2}g_{3}=\allowbreak
y^{6}+2y^{4}+4y^{3}=y^{2}\allowbreak g_{3}\twoheadrightarrow _{\left\{
g_{3}\right\} }0$.

When we iterate the process, $\left\{ g_{1},g_{2},g_{3},g_{4}\right\} $ is a
Gr\"obner basis for the ideal generated by $\left\{ g_{1},g_{2}\right\} $.
\end{example}

\section{Solving System of Algebraic Equations Over Finite Chain Rings \label%
{SolvingFCR}}

In this section, we assume as in Section \ref{GrobnerBases} that $R$ is a
finite commutative chain ring with maximal ideal $\mathfrak{m}$ generated by $\pi$, residue
field $\mathbb{F}_{q}=R/\mathfrak{m}$, and that $\nu$ is the nilpotency index of $\pi $.

\subsection{Solving With Lifting Approach}

The lifting approach consists of using solutions in the residue field $R/%
\mathfrak{m}$ to construct solutions in the ring $R$. One such method was
given in \cite[Algorithm 1]{Mikhailov2004solving}. Let us recall the main
property that justifies this algorithm. Consider a system of polynomial
equations of the form

\begin{equation}
f_{i}\left( x_{1},\ldots ,x_{k}\right) =0,\ \ i=1,\ldots ,d  \label{PolyEqua}
\end{equation}
where $f_{i}\left( x_{1},\ldots ,x_{k}\right) \in R\left[ x_{1},\ldots x_{k}%
\right] $. Set $\mathbf{f}\left( \mathbf{x}\right) =\left( f_{i}\left(
x_{1},\ldots ,x_{k}\right) \right) _{1\leq i\leq d}$, and \newline
$D\mathbf{f}\left( \mathbf{x}\right) =\left( D_{s}f_{i}\left( \mathbf{x}%
\right) \right) _{1\leq i\leq d,1\leq s\leq k}$, where $D_{s}f_{i}\left(
\mathbf{x}\right) $ is the partial derivative of $f_{i}\left( \mathbf{x}%
\right) $ with respect to $x_{s}$. Let $\overline{\mathbf{f}\left( \mathbf{x}%
\right) }$ be the canonical projection of $\mathbf{f}\left( \mathbf{x}%
\right) $ onto the residue field $\mathbb{F}_{q}$. As specified in Proposition %
\ref{DecompositionFCR}, every element $\mathbf{c}$ in $R^{k}$, has the $\pi-$adic decomposition $\mathbf{c}=\sum_{j=0}^{\nu -1}\pi ^{j}\gamma _{j}\left(
\mathbf{c}\right) $ where $\gamma _{j}\left( \mathbf{c}\right) \in $ $\Gamma
^{k}$. By \cite{Mikhailov2004solving} we have the following:

\begin{proposition}
\label{LiftingSolution}Assume that $\overline{\mathbf{f}\left( \mathbf{x}%
\right) }\neq \overline{\mathbf{0}}$. A vector $\mathbf{c}\in R^{k}$ is a
solution of (\ref{PolyEqua}) if and only if the vector $\gamma _{0}\left(
\mathbf{c}\right) $ is a solution in $\Gamma ^{k}$ of the system of
polynomial equations%
\begin{equation*}
\mathbf{f}\left( \mathbf{z}\right) \equiv \mathbf{0}\ \left( mod\ \pi
\right) ,
\end{equation*}%
and for $j\in \left\{ 1,\ldots ,\nu -1\right\} $ the vector $\gamma
_{j}\left( \mathbf{c}\right) $ is a solution in $\Gamma ^{k}$ of the system
of linear equations%
\begin{equation*}
D\mathbf{f}\left( \gamma _{0}\left( \mathbf{c}\right) \right) \cdot \mathbf{z%
}\equiv -\gamma _{j}\left( \mathbf{f}\left( \mathbf{c}^{[j]}\right) \right)
\ \ \left( mod\ \pi \right) .
\end{equation*}
\end{proposition}

By Proposition \ref{LiftingSolution}, solving systems of multivariate
polynomial equations over the finite chain ring $R$ is reduced to the finite
field $R/\mathfrak{m}$. However, the projection of certain systems of
equations on $R/\mathfrak{m}$ has several solutions and among these
solutions, only one can be lifted. For example, consider the following
system over $%
\mathbb{Z}
_{25}$:
\begin{equation}
\left\{
\begin{array}{c}
x^{5}-x=0 \\
5x+10=0%
\end{array}%
\right.  \label{ExSolveFCR1}
\end{equation}%
The residual field of $%
\mathbb{Z}
_{25}$ is $\mathbb{F}_{5}$. The projection of (\ref{ExSolveFCR1}) onto $%
\mathbb{F}_{5}$ is the equation $x^{5}-x=0$ and all elements in $\mathbb{F%
}_{5}$ are solutions. When we use Proposition \ref{LiftingSolution}, we
observe that only one solution can be lifted. Thus, the solution of (\ref%
{ExSolveFCR1}) in $%
\mathbb{Z}
_{25}$ is $x=18$. \ We can directly obtain this solution if we use the
method based on the canonical generating system \cite[Algorithm 2]%
{Mikhailov2004solving}. Thus, to solve systems of algebraic equations in one
variable in the next subsection, we will use the canonical generating system. For
the case of several variables, the methods which use solutions in the
residual field $R/\mathfrak{m}$ to construct solutions in the ring $R$ are
not appropriate in practice for some systems of polynomial equations,
specifically for parametric systems. As an illustration, consider the
following system over $%
\mathbb{Z}
_{8}$:%
\begin{equation}
\left\{
\begin{array}{c}
4x^{2}y+y^{3}+2y+4=0 \\
4xy^{2}=0%
\end{array}%
\right.  \label{ExSolveFCR2}
\end{equation}
Equation (\ref{ExSolveFCR2}) has $16$ solutions. When we use the lifting
approach to solve (\ref{ExSolveFCR2}) we compute each solution\textbf{\ }%
step by step. We will see in the next subsection that we can easily obtain
all these solutions using Gr\"obner bases.

\subsection{Solving With Gr\"obner Bases \label{SolvingWithGB1}}

In this subsection, we show how Gr\"obner bases can be used to solve
systems of multivariate polynomial equations over finite chain rings, as in
the case of finite fields. The direct consequence of Proposition \ref%
{EquivalenceGrobnerBases} is the elimination theorem given in \cite[Theorem
244]{Yengui2015constructive}.

\begin{proposition}
\label{EliminationTheorem}Let $G$ be a Gr\"obner basis for an ideal $I$ in
$R\left[ x_{1},\ldots ,x_{k}\right] $ with the lexicographic order $%
x_{1}>\cdots >x_{k}$. Then, for all $i$ in $\left\{ 1,\ldots ,k\right\} $, $%
G\cap R\left[ x_{i},\ldots ,x_{k}\right] $ is a Gr\"obner basis of $I\cap R%
\left[ x_{i},\ldots ,x_{k}\right] $.
\end{proposition}

The elimination theorem makes it possible to iteratively solve algebraic
systems by eliminating variables. Indeed, by Proposition \ref%
{EliminationTheorem}, if we compute a Gr\"obner basis $G$ of the ideal $%
I=\left( f_{1},\ldots ,f_{d}\right) $ associated to (\ref{PolyEqua})
with the lexicographic order $x_{1}>\cdots >x_{k}$, then $G$ will be of the
form $G=G_{1}\cup G_{2}\cup \cdots \cup G_{k}$, where $G_{1}=\left\{
g_{1,1}\left( x_{k}\right) ,\ldots ,g_{1,j_{1}}\left( x_{k}\right) \right\} $%
, \newline $G_{2}=\left\{ g_{2,1}\left( x_{k-1},x_{k}\right) ,\ldots ,g_{2,j_{2}}\left(
x_{k-1},x_{k}\right) \right\} $, $\ldots $, $G_{k}=\left\{ g_{k,1}\left(
x_{1},\ldots ,x_{k}\right) ,\ldots ,g_{k,j_{k}}\left( x_{1},\ldots
,x_{k}\right) \right\} $. So, (\ref{PolyEqua}) is equivalent to:

\begin{equation*}
\left\{
\begin{array}{c}
g_{1,1}\left( x_{k}\right) =\cdots =g_{1,j_{1}}\left( x_{k}\right) =0 \\
g_{2,1}\left( x_{k-1},x_{k}\right) =\cdots =g_{2,j_{2}}\left(
x_{k-1},x_{k}\right) =0 \\
\vdots \\
g_{k,1}\left( x_{1},\ldots ,x_{k}\right) =\cdots =g_{k,j_{k}}\left(
x_{1},\ldots ,x_{k}\right) =0%
\end{array}%
\right.
\end{equation*}

Thus, solving systems of multivariate polynomial equations is generally
reduced to solving systems of univariate polynomial equations. We will now
show how to use Gr\"obner bases over finite chain rings\ to solve systems
of univariate polynomial equations. Recall that a Gr\"obner basis $G$ is
called minimal if no proper subset of $G$ is a Gr\"obner basis for the
ideal generated by $G$. In \cite[Theorem 4.2]{Norton2001strong_cyclic}, a
characterization of minimal Gr\"obner bases in one variable over finite
chain rings has been given.

\begin{proposition}
\label{MinimalGrobnerBasis}Let $G\subset R[x]\backslash \left\{ 0\right\} $.
Then $G$ is a minimal Gr\"obner basis if and only if $G=\left\{ u_{0}\pi
^{a_{0}}g_{0},\ldots ,u_{s}\pi ^{a_{s}}g_{s}\right\} $ for some $0\leq s\leq
\nu -1$, $u_{i}\in R$ and $g_{i}\in R[x]$ for $i=0,\ldots ,s$ and such that:
\begin{enumerate}
    \item[(i)] $0\leq a_{0}<a_{1}<\cdots
<a_{s}\leq \nu -1$ and for $i=0,\ldots ,s$, $u_{i}$ is a unit;

\item[(ii)] for $i=0,\ldots ,s$, $g_{i}$ is monic;

\item[(iii)] $\deg \left( g_{i}\right) >\deg \left( g_{i+1}\right) $ for any $i \in
\{0,\ldots ,s-1\}$;

\item[(iv)] for $i=0,\ldots ,s-1$, $\pi ^{a_{i+1}}g_{i}$ is in the ideal generated by $\left\{ \pi
^{a_{i+1}}g_{i+1},\ldots ,\pi ^{a_{s}}g_{s}\right\} $.
\end{enumerate}
\end{proposition}

As specified in \cite{Mikhailov2004solving}, a minimal Gr\"obner basis in
one variable over finite chain rings is a canonical generating system.
Therefore, according to Proposition \ref{MinimalGrobnerBasis}, we can use
\cite[Algorithm 2]{Mikhailov2004solving} to solve systems of univariate
polynomial equations over finite chain rings using Gr\"obner bases.
Specifically, consider the system of univariate polynomial equations of the
form
\begin{equation}
f_{i}\left( x\right) =0,\ \ \ \ \ i=1,\ldots ,r  \label{PolyEqua1var1}
\end{equation}%
where $f_{i}\left( x\right) \in R\left[ x\right] $. Assume that a minimal Gr\"obner basis of the ideal generated by $\left\{ f_{1}\left( x\right) ,\ldots
,f_{r}\left( x\right) \right\} $ is $G=\left\{ u_{0}\pi ^{a_{0}}g_{0},\ldots
,u_{s}\pi ^{a_{s}}g_{s}\right\} $ like in Proposition \ref%
{MinimalGrobnerBasis}. As specified in \cite[page 64]{Mikhailov2004solving}
we can assume that $a_{0}=0$. Set $h_{j}=g_{i}$, for $0\leq i\leq s$ and $%
a_{i}\leq j<$ $a_{i+1}$, where $a_{s+1}=\nu $. Then, Equation (\ref%
{PolyEqua1var1}) is equivalent to the following system of polynomial
equations:

\begin{equation}
\pi ^{j}h_{j}\left( x\right) =0,\ \ \ \ j=0,\ldots ,\nu-1.
\label{PolyEqua1var2}
\end{equation}

As in \cite[Theorem 8]{Mikhailov2004solving} and \cite[Equation (54)]%
{Mikhailov2004solving}, we will use the derivation $Dh_{j}\left( x\right) $
of $h_{j}\left( x\right) $ to solve Equation (\ref{PolyEqua1var2}). Recall
that any element $c$ in $R$ admits a unique $\pi -$adic decomposition $%
c=\sum_{j=0}^{\nu -1}\gamma _{j}\left( c\right) \pi ^{j}$, where $\gamma
_{j}\left( c\right) \in $ $\Gamma $.

\begin{proposition}
\label{LiftingSolutionOneVariable}An element $c$ in $R$, is a solution of (%
\ref{PolyEqua1var2}) if and only if $\gamma _{0}\left( c\right) $ is a
solution in $\Gamma $ of the polynomial equation%
\begin{equation*}
h_{\nu -1}\left( x\right) \equiv 0\ \left( mod\ \pi \right) ,
\end{equation*}%
and for $j\in \left\{ 1,\ldots ,\nu -1\right\} $, $\gamma _{j}\left(
c\right) $ is a solution in $\Gamma $ of the linear equation:%
\begin{equation*}
Dh_{\nu -j-1}\left( \gamma _{0}\left( c\right) \right) x\equiv -\gamma
_{j}\left( h_{\nu -j-1}\left( c^{[j]}\right) \right) \ \ \left( mod\ \pi
\right) .
\end{equation*}
\end{proposition}

According to Propositions \ref{EliminationTheorem} and \ref%
{LiftingSolutionOneVariable}, to solve a system of multivariate polynomial
equations over finite chain rings, we can compute a Gr\"obner basis of the
associated system with the lexicographic order and successively solve
systems of univariate polynomial equations. We will see in Sections \ref%
{MinRankproblem} and \ref{RDproblem} that this approach is appropriate for
some systems of algebraic equations when we just need a partial solution.

\begin{example}
\label{ExSolveFCR_GB} Let us solve Equation (\ref{ExSolveFCR2}) over $%
\mathbb{Z}
_{8}$ using Gr\"obner bases. According to Example \ref{ExGB}, a Gr\"obner basis
with the lexicographic order $x>y$ of the ideal $I$ generated by $\left\{
4x^{2}y+y^{3}+2y+4,4xy^{2}\right\} $ is $G=\left\{
g_{1,1},g_{1,2},g_{2,1},g_{2,2}\right\} $ where $g_{1,1}\left( y\right)
=y^{4}+2y^{2}+4y$, $g_{1,2}\left( y\right) =2y^{3}+4y$, $g_{2,1}\left(
x,y\right) =4x^{2}y+y^{3}+2y+4$, $g_{2,2}\left( x,y\right) =4xy^{2}$. By
Proposition \ref{EliminationTheorem}, a Gr\"obner basis of $I\cap R\left[ y%
\right] $ is $G_{1}=G\cap R\left[ y\right] =\left\{ g_{1,1}\left( y\right)
,g_{1,2}\left( y\right) \right\} $. So, we can use $G_{1}$ to find the
partial solution $y$ of (\ref{ExSolveFCR2}). The system
\begin{equation*}
g_{1,1}\left( y\right) =g_{1,2}\left( y\right) =0
\end{equation*}%
is equivalent to%
\begin{equation}
h_{1,1}\left( y\right) =2h_{1,2}\left( y\right) =4h_{1,3}\left( y\right) =0
\label{ExSolveFCR2_y}
\end{equation}%
where $h_{1,1}\left( y\right) =g_{1,1}\left( y\right) $ and $h_{1,2}\left(
y\right) =h_{1,3}\left( y\right) =y^{3}+2y$. Let $c$ be a solution of (\ref%
{ExSolveFCR2_y}). We have $c=\gamma _{0}\left( c\right) +2\gamma _{1}\left(
c\right) +4\gamma _{2}\left( c\right) $ where $\gamma _{j}\left( c\right)
\in \Gamma =\left\{ 0,1\right\} $ for $j\in \left\{ 0,1,2\right\} $. By
Proposition \ref{LiftingSolutionOneVariable}, $\gamma _{0}\left( c\right) $
is a solution in $\Gamma $ of the equation $h_{1,3}\left( c\right) \equiv 0\
\left( mod\ 2\right) $. So, $\gamma _{0}\left( c\right) =0$. By Proposition %
\ref{LiftingSolutionOneVariable}, $\gamma _{1}\left( c\right) $ is a
solution in $\Gamma $ of the equation $Dh_{1,2}\left( \gamma _{0}\left(
c\right) \right) y\equiv -\gamma _{1}\left( h_{1,2}\left( c^{[1]}\right)
\right) \ \left( mod\ 2\right) $. We have $c^{[1]}=\gamma _{0}\left(
c\right) =0$, $h_{1,2}\left( c^{[1]}\right) =0$, $\gamma _{1}\left(
h_{1,2}\left( c^{[1]}\right) \right) =0$, $Dh_{1,2}\left( y\right) =3y^{2}+2$%
, and $Dh_{1,2}\left( \gamma _{0}\left( c\right) \right) =2$. Therefore, $%
\gamma _{1}\left( c\right) $ is a solution of\ $2y\equiv 0\ \left( mod\
2\right) $. Thus, $\gamma _{1}\left( c\right) \in \left\{ 0,1\right\} $.
Using the same reasoning, for $\gamma _{1}\left( c\right) =0$ or $\gamma
_{1}\left( c\right) =1$, we compute $\gamma _{2}\left( c\right) \in \left\{
0,1\right\} $. Therefore, $c\in \left\{ 0,2,4,6\right\} $. Thus, the partial
solution $y$ of (\ref{ExSolveFCR2}) is in $\left\{ 0,2,4,6\right\} $. To
find the partial solution $x$ corresponding for example to $y=0$, we must
first calculate a Gr\"obner basis of $\left\{ g_{2,1}\left( x,0\right)
,g_{2,2}\left( x,0\right) \right\} $. But for all $x$ in $%
\mathbb{Z}
_{8}$, $g_{2,1}\left( x,0\right) =4\neq 0$, $g_{2,1}\left( x,4\right) =4\neq
0$, $g_{2,1}\left( x,2\right) =g_{2,2}\left( x,2\right) =0$, and $%
g_{2,1}\left( x,6\right) =g_{2,2}\left( x,6\right) =0$. Thus, $y$ is in $%
\left\{ 2,6\right\} $ and the solutions of (\ref{ExSolveFCR2}) are $\left\{
\left( t,2\right) ,\left( t,6\right) ,t\in
\mathbb{Z}
_{8}\right\} $.
\end{example}

\begin{remark}
\label{RingEqua}In \cite[Theorem 5.14]{Nechaev2008finite}, the monic
polynomial $F_{m}$ with smaller degree and satisfying 
$%
F_{m}\left( x\right) =0$ for all $x$ in $R$, has been defined. Thus, as in the case of finite
fields, to simplify the resolution of (\ref{PolyEqua}) we can add the following equations $%
F_{m}\left( x_{1}\right) =\cdots =F_{m}\left( x_{k}\right) =0$.
\end{remark}
\begin{example}
\label{ExSolveFCR_GB2}Consider again Equation (\ref{ExSolveFCR2}) over $%
\mathbb{Z}
_{8}$. A Gr\"obner basis with the lexicographic order $y>x$ of the ideal $I$
generated by $\left\{ 4x^{2}y+y^{3}+2y+4,4xy^{2}\right\} $ is the same set
\newline
$\left\{ 4x^{2}y+y^{3}+2y+4,4xy^{2}\right\} $. Consequently, $I\cap
\mathbb{Z}
_{8}[x]=\left\{ 0\right\} $. So, we cannot solve Equation (\ref{ExSolveFCR2}%
) directly by using only Proposition \ref{LiftingSolutionOneVariable} with
the lexicographic order $y>x$. However, according to \cite[Theorem 5.14]%
{Nechaev2008finite}, the monic polynomial $F_{m}$ for the ring $%
\mathbb{Z}
_{8}$ is defined by
\begin{equation*}
F_{m}\left( x\right) =\left( x^{2}-x\right) ^{2}-2\left( x^{2}-x\right) .
\end{equation*}

A Gr\"obner basis with the lexicographic order $y>x$ of the ideal generated by
\newline
$\left\{ 4x^{2}y+y^{3}+2y+4,4xy^{2},F_{m}\left( x\right),F_{m}\left(
y\right) \right\} $ is $\left\{ y^{2}+4,2y+4,F_{m}\left( x\right) \right\} $%
. Therefore, (\ref{ExSolveFCR2}) is equivalent to $y^{2}+4=2y+4=0$. We solve
the system $y^{2}+4=2y+4=0$ using Proposition \ref%
{LiftingSolutionOneVariable}, and we obtain $y=2$ or $y=6$. Thus, the
solutions of (\ref{ExSolveFCR2}) are $\left\{ \left( t,2\right) ,\left(
t,6\right) ,t\in
\mathbb{Z}
_{8}\right\} $.
\end{example}

\section{Solving Systems of Algebraic Equations Over Finite Commutative
Local Rings \label{LocalRings}}

In the preview section, we have used Gr\"obner bases to show how one can solve
systems of algebraic equations over finite chain rings. We will now show
that solving systems of algebraic equations over finite commutative rings
can be reduced to finite chain rings. According to \cite[ Theorem VI.2]%
{Mcdonald1974finite}, if $R$ is a finite commutative ring, then $R$ can be
decomposed as a direct sum of local rings, that is to say $R \cong
R_{(1)}\times \cdots \times R_{(\rho )}$ where for $j=1,\ldots ,\rho $, $R_{(j)}$ is a finite
commutative local ring. Thus, the problem of solving systems of algebraic equations over $R$ can be reduced to solving systems of algebraic equations over the various $R_{(j)}$. However, Gr\"o%
bner basis are not generally equal to strong Gr\"obner bases over local
rings. Therefore, we will use Galois rings, which are specific classes of finite chain rings to represent finite local rings. As specified in \cite%
{Agrawal2005automorphisms,Behboodi2014classification},\ finite rings have
several representations (the table representation, the basis
representation, and the polynomial representation). Galois rings can be used
to give the basis representation and the polynomial representation of \
finite commutative local rings \cite[Theorems XVI.2 and XVII.1]%
{Mcdonald1974finite}. In \cite{Bulyovszky2017polynomial}, Bulyovszky and Horv\'ath used the basis representation to give a good method for solving systems
of linear equations over finite local rings. We are going to extend this method to
systems of multivariate polynomial equations.

In this section, we assume that $R$ is a finite commutative local ring with
maximal ideal $\mathfrak{m}$ and residue field $\mathbb{F}_{q}=R/%
\mathfrak{m}$. Set $q=p^{\mu }$ where $p$ is a prime number. Then the
characteristic of $R$ is $p^{\varsigma }$ where $\varsigma $ is a
non-negative integer and by \cite[ Theorem XVII.1]{Mcdonald1974finite}\
there is a subring $R_{0}$ of $R$ such that $R_{0}$ is isomorphic to the
Galois ring of characteristic $p^{\varsigma }$ and cardinality $p^{\mu
\varsigma }$. Considering $R$ as a $R_{0}-$module, there exist $\theta
_{1},\ldots ,\theta _{\gamma }$ in $R$ such that
\begin{equation}
R=R_{0}\theta _{1}\oplus \cdots \oplus R_{0}\theta _{\gamma }.
\label{BasisRepresentation}
\end{equation}%
Let $j$ in $\left\{ 1,\ldots ,\gamma \right\} $. Since every ideal in $R_{0}$
is generated by a power of $p$, then there is $\varsigma _{j}$ in $\left\{
1,\ldots ,\varsigma \right\} $ such that
\begin{equation*}
p^{\varsigma _{j}}R_{0}=Ann\left( \theta _{j}\right) =\left\{ a\in
R_{0}:a\theta _{j}=0\right\} .
\end{equation*}

\begin{lemma}
\label{EquivLocal}Let $u$ in $R$ and $u_{j}$ in $R_{0}$ such that $%
u=\sum_{j=1}^{\gamma }u_{j}\theta _{j}$. The following statements are
equivalent:

(a) $u=0$;

(b) for all $j\in \left\{ 1,\ldots ,\gamma \right\} $, $\ \theta _{j}u_{j}=0$%
;

(c) for all $j\in \left\{ 1,\ldots ,\gamma \right\} $, $\ p^{\varsigma
-\varsigma _{j}}u_{j}=0$.

Moreover, each element $u_{j}$ is unique modulo $p^{\varsigma _{j}}$.
\end{lemma}

Lemma \ref{EquivLocal}\ and the basis decomposition (\ref%
{BasisRepresentation}) can be used to transform a system of multivariate
polynomial equations over finite local rings to Galois rings. Specifically,
we have the following:

\begin{theorem}
\label{EquiEquaLocalRing}Consider a system of polynomial equations of the
form

\begin{equation}
f_{r}\left( \left( x_{i}\right) _{1\leq i\leq k}\right) =0,\ \ r=1,\ldots ,d
\label{EquaLocalRing1}
\end{equation}
where $f_{r}$ are multivariate polynomial functions with coefficients in $R$
and $\left( x_{i}\right) _{1\leq i\leq k}\in R^{k}$. Set
\begin{equation*}
x_{i}=\sum_{j=1}^{\gamma }x_{i,j}\theta _{j},\ \ i=1,\ldots ,k
\end{equation*}%
where $x_{i,j}\in R_{0}$ and
\begin{equation*}
f_{r}\left( \left( x_{i}\right) _{1\leq i\leq k}\right) =\sum_{s=1}^{\gamma
}f_{r,s}\left( \left( x_{i,j}\right) _{1\leq i\leq k,1\leq j\leq \gamma
}\right) \theta _{s},\ \ r=1,\ldots ,d
\end{equation*}%
where $f_{r,s}$ are multivariate polynomial functions with coefficients in $%
R_{0}$. Then Equation (\ref{EquaLocalRing1}) is equivalent to
\begin{equation}
p^{\varsigma -\varsigma _{s}}f_{r,s}\left( \left( x_{i,j}\right) _{1\leq
i\leq k,1\leq j\leq \gamma }\right) =0,\ \ r=1,\ldots ,d,\ s=1,\ldots
,\gamma .  \label{EquaLocalRing2}
\end{equation}
\end{theorem}

Since Galois rings are specific cases of finite chain rings, we can use the
methods described in Section \ref{SolvingFCR} to solve (\ref{EquaLocalRing2}%
).

\begin{example}
\label{ExSolve_LR}In this example we consider a local ring of size $16$
which in not a finite chain ring. As specified in \cite{Martinez2015codes},
we can choose $R=%
\mathbb{Z}
_{8}\left[ X\right] /I$ where $I$ is the ideal generated by $X^{2}+4$ and $%
2X$. Then $R$ is a local ring with maximal ideal generated by $2+I$
and $X+I$. Set $\theta =X+I$, then a maximal Galois subring of $R$ is $R_{0}=%
\mathbb{Z}
_{8}$ and we have $R=\theta _{1}R_{0}\oplus \theta _{2}R_{0}$ where $\theta
_{1}=1$ and $\theta _{2}=\theta $. Moreover, $Ann\left( \theta _{1}\right)
=\left\{ 0\right\} =2^{3}R_{0}$ and $Ann\left( \theta _{2}\right) =2R_{0}$.
We would like to find the roots of the polynomial function defined over $R$
by
\begin{equation*}
P\left( x\right) =x^{3}+2x+4.
\end{equation*}%
The residual field of $R$ is $\mathbb{F}_{2}$ and the projection over $%
\mathbb{F}_{2}$ \ of $P\left( x\right) $ is $\overline{P}\left( x\right)
=x^{3}$ which is not square-free. Therefore, we are not able to find the
roots of $P$ using methods based on the Hensel's lemma \cite[Theorem XIII.4]%
{Mcdonald1974finite} or the Newton-Hensel's lemma \cite[Proposition 2.1.9]%
{Fontein2005elliptic}. Thus, an alternative method is to use Theorem \ref%
{EquiEquaLocalRing}. Set $x=x_{1}+x_{2}\theta $ where $x_{1}$ and $x_{2}$
are in $R_{0}$. Then,%
\begin{equation*}
P\left( x_{1}+x_{2}\theta \right) =x_{1}^{3}+4x_{1}x_{2}^{2}+2x_{1}+4+\theta
x_{1}^{2}x_{2}.
\end{equation*}%
Therefore, the equation
\begin{equation*}
x^{3}+2x+4=0
\end{equation*}%
is equivalent to the system
\begin{equation}
\left\{
\begin{array}{c}
x_{1}^{3}+4x_{1}x_{2}^{2}+2x_{1}+4=0 \\
4x_{1}^{2}x_{2}=0%
\end{array}%
\right.  \label{EquaLocalRing3}
\end{equation}

Thanks to Example \ref{ExSolveFCR_GB}, we deduce that the solutions of (\ref{EquaLocalRing3}) are $%
\left( x_{1},x_{2}\right) $ in $\left\{ \left( 2,t\right) ,\left( 6,t\right)
,t\in
\mathbb{Z}
_{8}\right\} $. As $2\theta =0$ and $x=x_{1}+x_{2}\theta $, then $x_{2}$ is
unique modulo $2$. We can therefore choose $x_{2}$ in $\left\{ 0,1\right\} $%
. Thus, the roots of $P$ are $2$, $6$, $2+\theta $, and $6+\theta $.
\end{example}

\section{MinRank Problem Over Finite Principal Ideal Rings \label%
{MinRankproblem}}

In this section, we extend some algebraic modeling of the MinRank problem
over finite principal ideal rings. In what follow, we assume that $R$ is
a finite commutative principal ideal ring. The set of all $m\times n$
matrices with entries in a ring $R$ will be denoted by $R^{m\times n}$. Let $%
\mathbf{A}\in R^{m\times n}$, we denote by $row\left( \mathbf{A}\right) $
the $R-$submodule generated by the row vectors of $\mathbf{A}$. The
transpose of $\mathbf{A}$ is denoted by $\mathbf{A}^{\top }$ and the $%
k\times k$ identity matrix is denoted by $\mathbf{I}_{k}$.

\subsection{MinRank Problem}
\begin{definition}
Let $\ \mathbf{A}\in R^{m\times n}$. The \textbf{rank} of $\mathbf{A}$,
denoted by $rk_{R}\left( \mathbf{A}\right) $ or simply by $rk\left( \mathbf{A%
}\right) $ is the smallest number of elements in $row(\mathbf{A})$ which
generate $row(\mathbf{A})$ as a $R-$module.
\end{definition}

As specified in \cite[Proposition 3.4]{Kamche2019rank}, the Smith normal
form can be used to compute the rank of a matrix. Moreover, as in the case
of fields, the map $R^{m\times n}\times R^{m\times n}\rightarrow
\mathbb{N}
$, given by $\left( \mathbf{A,B}\right) \mapsto row\left( \mathbf{A-B}%
\right) $ is a metric. However, some properties of the rank of a matrix over
fields generally do not extend to rings due to zero divisors.

\begin{example}
Consider the matrix $\mathbf{A=}\left(
\begin{array}{cc}
2 & 0 \\
0 & 4%
\end{array}%
\right) $ over $\mathbb{Z}_{8}$. Then, $rk\left( \mathbf{A}\right) =2$,  $%
rk\left( 6\mathbf{A}\right) =1$ and $det(\mathbf{A})=0$. Thus, $rk\left(
\mathbf{A}\right) \neq rk\left( 6\mathbf{A}\right) $ and $rk\left( \mathbf{A}%
\right) $  is not equal to the order of the highest order non-vanishing
minor.
\end{example}

\begin{definition}
Let $\mathbf{M}_{0}$, $\mathbf{M}_{1}$, $\ldots $, $\mathbf{M}_{k}$ in $%
R^{m\times n}$ and $r$ in ${{\mathbb{N}}}^{\ast }$. The \textbf{MinRank
problem} is to find $x_{1},\ldots ,x_{k}$ in $R$ such that $rk(\mathbf{M}%
_{0}+\sum_{i=1}^{k}x_{i}\mathbf{M}_{i})\leq r$. The \textbf{homogeneous
MinRank problem} corresponds to the case where $\mathbf{M}_{0}=\mathbf{0}$.
\end{definition}

Let $\mathbf{x}\in R^{n}$, and $\mathbf{D}_{x}$ be the $n \times n$ diagonal matrix with the entries of $\mathbf{x}$ on the diagonal. Then, according to \cite[Proposition 3.4%
]{Kamche2019rank}, the Hamming weight of $\mathbf{x}$ is equal to the rank of
$\mathbf{D}_{x}$ and thus, using the work of \cite{Weger2020hardness}, one can
easily prove as in \cite{Courtois2001efficient} that the MinRank problem over
finite principal ideal rings is NP-complete. Therefore, this problem is potential for cryptography and, the study of its algebraic resolution deserves attention. In  \cite[Proposition 4.3]{Kalachi2023OnTheRank}, it was shown that solving
the rank decoding problem over finite principal ideal rings reduces to
finite chain rings. This result does apply also to the MinRank problem.

In general, an instance of the MinRank problem has several solutions.
But if $r$ is not greater than the error correction capability of
the $R-$linear code generated by $\mathbf{M}_{1},\ldots ,\mathbf{M}_{k}$ (assuming $\mathbf{M}_{1},\ldots ,\mathbf{M}_{k}$ are $R-$linearly independent), then the problem has a unique solution $\left(
x_{1},\ldots ,x_{k} \right)$.
In the homogeneous case, for any solution  $\left( x_{1},\ldots ,x_{k}\right) $ and for any $\alpha \in R$,  $\left( \alpha x_{1},\ldots ,\alpha x_{k}\right) $ is
also a solution. Thus, if $R$ is a field, one of
the components of a nonzero solution of the homogeneous MinRank problem can always
be assumed to be $1$. However, if $R$ is not a field, this
assumption is not true in some cases (see Example \ref{Ex1KSModeling}).

From a modelling perspective, the MinRank problem over finite fields can be transformed into a system of
algebraic equations using the maximum minors while over finite principal
ideal rings, the rank of a matrix is usually not equal to the order of the
highest order non-vanishing minor. As a consequence, the MaxMinor modelling does not apply in general
when dealing with rings. In the following subsections, we will prove that the
Kipmis-Shamir Modelling and the Support Minors Modelling can be extended over
finite principal ideal rings.
\subsection{Kipmis-Shamir Modeling}

We start with some lemmas which will be used to give the Kipmis-Shamir
modeling over finite principal ideal rings. According to \cite[Proposition
3.2]{Kamche2019rank}, we have the following:

\begin{lemma}
\label{Envelope}
Let $\mathbf{E} \in R^{m\times n}$ such that $rk\left(
\mathbf{E}\right) \leq r$. Then, there exists a rank $r$ free submodule $F$ of $R^{n}$ such that $row\left( \mathbf{E}\right) \subset F$.
\end{lemma}

\begin{remark}
\label{NbEnvelope}Let $\mathbf{E}$ and $F$ as in Lemma \ref{Envelope}. If $%
row\left( \mathbf{E}\right) $ is a free module and $rk\left( \mathbf{E}%
\right) =r$ then $F$ is unique and $row\left( \mathbf{E}\right) =F$. But if $%
row\left( \mathbf{E}\right) $ is not a free module, then $F$ is generally not
unique.
\end{remark}

\begin{example}
Consider the matrix $\mathbf{E=}\left(
\begin{array}{ccc}
2 & 0 & 4%
\end{array} \right) $ over $
\mathbb{Z}_{8}$. Then $rk\left( \mathbf{E}\right) =1$ and there exist four free
submodules $F$ of $%
\mathbb{Z}
_{8}^{3}$ of rank $1$ such that $row\left( \mathbf{E}\right) \subset F$.
These four submodules are respectively generated by $\left( 1,0,2\right) $, $%
\left( 1,4,2\right) $, $\left( 1,0,6\right) $,and $\left( 1,4,6\right) $.
\end{example}

By \cite[Proposition 2.9]{Fan2014matrix}, we have the following:

\begin{lemma}
\label{ParityCheckMat} For any rank $r$ free submodule $F$ of $R^{n}$,
there exists $\mathbf{Z}\in R^{n\times (n-r)}$ with linearly independent column vectors and satisfying 
\[
\forall~ \mathbf{y} \in R^{n},~ \mathbf{y}
\in F~ \Longleftrightarrow ~\mathbf{yZ}=\mathbf{0}.
\]
\end{lemma}

If $a$ and $b$ are two elements of a finite chain ring, then $a$ divides $b$
or $b$ divides $a$. This property was used in \cite[Proposition 3.2]%
{Norton2000structure} to prove the existence of the generator matrices in
standard form over finite chain rings. So, we have the following:

\begin{lemma}
\label{DecomFCR}Assume that $R$ is a finite chain ring. Let $\mathbf{Z}\in
R^{n\times (n-r)}$ with column vectors that are linearly independent.
Then there exists a size $n$ permutation matrix $\mathbf{P}$, an invertible
matrix $\mathbf{Q}\in R^{(n-r)\times (n-r)}$, and a matrix $\mathbf{Z}%
^{\prime }\in R^{r\times (n-r)}$ such that
\begin{equation*}
\mathbf{Z}=\mathbf{P}\left(
\begin{array}{c}
\mathbf{I}_{n-r} \\
\mathbf{Z}^{\prime }%
\end{array}%
\right) \mathbf{Q}.
\end{equation*}
\end{lemma}

The above Lemma \ref{DecomFCR} is not generally true when $R$ is not a finite chain
ring. Indeed, consider the matrix
\begin{equation*}
\mathbf{Z=}\left(
\begin{array}{c}
2 \\
3%
\end{array}%
\right)
\end{equation*}
over $%
\mathbb{Z}
_{6}$. The column vector of \ $\mathbf{Z}$ is $%
\mathbb{Z}
_{6}-$linearly independent. But $\mathbf{Z}$ cannot be decomposed as in
Lemma \ref{DecomFCR}.
Lemmas \ref{Envelope}, \ref{ParityCheckMat} and \ref{DecomFCR} allow\ to
extend the Kipmis-Shamir Modeling to finite principal ideal rings.

\begin{theorem}
\label{KipmisShamirModeling}Let $\mathbf{M}_{0}$, $\mathbf{M}_{1},\ldots ,%
\mathbf{M}_{k}$ in $R^{m\times n}$, $x_{1},\ldots ,x_{k}$ in $R$ and $r$ in $%
{{\mathbb{N}}}^{\ast }$. 
For $M_{x}=\mathbf{M}_{0}+\sum_{i=1}^{k}x_{i}%
\mathbf{M}_{i}$, the following statements are equivalent.
\begin{itemize}
\item[(i)] $rk(\mathbf{M}_{x})\leq r$.

\item[(ii)] There exists $\mathbf{Z}\in R^{n\times (n-r)}$, with column vectors that are
linearly independent and such that
\begin{equation}
\mathbf{M}_{x}\mathbf{Z=0}.  \label{KSModelingEq}
\end{equation}%
\end{itemize}
Moreover, if $R$ is a finite chain ring then, up to a permutation of
columns of $\mathbf{M}_{x}$, we can assume that $\mathbf{Z}$ is into the form
\begin{equation*}
\mathbf{Z=}\left(
\begin{array}{c}
\mathbf{I}_{n-r} \\
\mathbf{Z}^{\prime }%
\end{array}%
\right)
\end{equation*}%
where $\mathbf{Z}^{\prime }\in R^{r\times (n-r)}$.
\end{theorem}

\begin{proof}
The proof is similar to the case of fields. Indeed, assume that $rk(\mathbf{M%
}_{x})\leq r$. Then, by Lemma \ref{Envelope}, there exists a free submodule $%
F$ of $R^{n}$ of rank $r$ such that $row\left( \mathbf{M}_{x}\right) \subset
F$. Thus, by Lemma \ref{ParityCheckMat}, there is $\mathbf{Z}\in R^{n\times
(n-r)}$, with column vectors that are linearly independent and such that (\ref%
{KSModelingEq}) holds.
Conversely, assume that (ii) holds. Then, by Lemma \ref{ParityCheckMat}, all row
vectors of $\mathbf{M}_{x}$ are in a free module of rank $r$. Therefore, by
\cite[Proposition 3.2]{Kamche2019rank}, $rk(\mathbf{M}_{x})\leq r$.
\end{proof}

As specified in Remark \ref{NbEnvelope}, the free submodule $F$ is generally
not unique. Therefore, $\mathbf{Z}^{\prime }$ is generally not unique.

\begin{example}
\label{Ex1KSModeling} Consider the following MinRank problem: find $x_{1}$, $x_{2}$, $x_{3}$ in   $\mathbb{Z}_{8}$ such that 
\begin{equation}
rk\left( x_{1}\mathbf{M}_{1}+x_{2}\mathbf{M}_{2}+x_{3}\mathbf{M}_{3}\right)
\leq 1  \label{Ex1KSModeling1}
\end{equation}%
where
\begin{equation*}
M_{1}=\left(
\begin{array}{rrrr}
0 & 0 & 0 & 7 \\
1 & 0 & 0 & 5 \\
0 & 1 & 0 & 2 \\
0 & 0 & 1 & 4%
\end{array}%
\right),\ M_{2}=\left(
\begin{array}{rrrr}
0 & 0 & 7 & 4 \\
0 & 0 & 5 & 3 \\
1 & 0 & 2 & 5 \\
0 & 1 & 4 & 2%
\end{array}%
\right),\ M_{3}=\left(
\begin{array}{rrrr}
2 & 2 & 0 & 4 \\
4 & 2 & 0 & 6 \\
0 & 4 & 2 & 4 \\
0 & 6 & 6 & 0%
\end{array}%
\right).
\end{equation*}%
Since $r=1$, by Theorem \ref{KipmisShamirModeling}, (\ref%
{Ex1KSModeling1}) is equivalent to
\begin{equation}
\left( x_{1}\mathbf{M}_{1}+x_{2}\mathbf{M}_{2}+x_{3}\mathbf{M}_{3}\right)
\left(
\begin{array}{ccc}
1 & 0 & 0 \\
0 & 1 & 0 \\
0 & 0 & 1 \\
z_{1} & z_{2} & z_{3}%
\end{array}%
\right) =\mathbf{0}  \label{Ex1KSModeling2}
\end{equation}%
A Gr\"obner basis associated to (\ref{Ex1KSModeling2}) with the lexicographic
order $z_{1}>z_{2}$ $>z_{3}>x_{1}>x_{2}>x_{3}$ is $2z_{1}x_{3}+6x_{3}$, $%
2z_{2}x_{3}+6x_{3}$, $2z_{3}x_{3}+6x_{3}$, $x_{1}+2x_{3}$, $x_{2}+2x_{3}$, $%
4x_{3}$. According to Proposition \ref{LiftingSolutionOneVariable}, the
solutions of the system $x_{1}+2x_{3}=x_{2}+2x_{3}=4x_{3}=0$ are $\left(
x_{1},x_{2},x_{3}\right) \in \left\{ \left( 0,0,0\right) ,\left(
4,4,2\right) ,\left( 0,0,4\right) ,\left( 4,4,6\right) \right\} $. Thus, (\ref{Ex1KSModeling1})
has exactly four solutions.
\end{example}

In the simulations, we observe that in some cases to simplify the resolution of (%
\ref{KSModelingEq}) it is necessary to add some equations as specified in
Remark \ref{RingEqua}.

\begin{example}
\label{Ex2KSModeling} Consider the following MinRank problem: find $x_{1}$, $x_{2}$, $x_{3}$ in   $\mathbb{Z}_{8}$ such that
\begin{equation}
rk(\mathbf{M}_{x})\leq 1  \label{Ex2KSModeling1}
\end{equation}%
where $\mathbf{M}_{x}=\mathbf{M}_{0}+\sum_{i=1}^{3}x_{i}\mathbf{M}_{i}$ and%
\begin{equation*}
\mathbf{M}_{0}=\left(
\begin{array}{rrr}
5 & 2 & 3 \\
5 & 1 & 4 \\
4 & 3 & 6%
\end{array}%
\right) ,\ \ \ \ \mathbf{M}_{1}=\left(
\begin{array}{rrr}
1 & 2 & 0 \\
0 & 1 & 3 \\
0 & 2 & 1%
\end{array}%
\right)
\end{equation*}%
\newline
\begin{equation*}
\mathbf{M}_{2}=\left(
\begin{array}{rrr}
0 & 2 & 1 \\
1 & 0 & 3 \\
0 & 5 & 5%
\end{array}%
\right) ,\ \ \ \ \mathbf{M}_{3}=\left(
\begin{array}{rrr}
0 & 5 & 5 \\
0 & 1 & 0 \\
1 & 2 & 5%
\end{array}%
\right)
\end{equation*}%
\newline
According to Theorem \ref{KipmisShamirModeling}, (\ref{Ex2KSModeling1}) is
equivalent to
\begin{equation}
\mathbf{M}_{x}\mathbf{Z}=\mathbf{0}.  \label{Ex2KSModeling2}
\end{equation}%
When we choose $\mathbf{Z}$ in the form
\begin{equation*}
\mathbf{Z=}\left(
\begin{array}{cc}
1 & 0 \\
0 & 1 \\
z_{1} & z_{2}%
\end{array}%
\right)
\end{equation*}%
then we do not get the solution. Thus, it is necessary to choose the
switchable permutation. In our simulations, we observed that we can choose
\begin{equation*}
\mathbf{Z=}\left(
\begin{array}{cc}
z_{1} & z_{2} \\
1 & 0 \\
0 & 1%
\end{array}%
\right).
\end{equation*}%
In this case, when we compute a Gr\"obner basis associated to (\ref%
{Ex2KSModeling2}) with the lexicographic order $z_{1}>z_{2}$ $%
>x_{1}>x_{2}>x_{3}$, we obtain a system which is not easy to solve. But when
we add the polynomial expressions $F_{m}\left( z_{1}\right) $, $F_{m}\left(
z_{2}\right) $, $F_{m}\left( x_{1}\right) $, $F_{m}\left( x_{2}\right) $, $%
F_{m}\left( x_{3}\right) $ as in Example \ref{ExSolveFCR_GB2}, we get the Gr\"obner basis $z_{1}^{4}-z_{1}^{2}$, $2z_{1}$, $z_{2}^{4}+3z_{2}^{2}+4$, $%
2z_{2}+4$, $x_{1}+7$, $x_{2}+5$, $x_{3}+2$. Thus, we directly obtain the
solution of (\ref{Ex2KSModeling1}) which is $x_{1}=1$, $x_{2}=3$, $x_{3}=6$.
\end{example}

\subsection{Support-Minors Modeling}

In this subsection, we will show that the Support-Minors modeling of the
MinRank problem given in \cite{BBBGNRT20} over fields can be extended to
finite principal ideal rings.

\begin{lemma}
\label{MaxMinLem} Let $\mathbf{A}\in R^{r\times n}$ with row vectors that
are linearly independent, and $\mathbf{y\in }R^{n}$. Then $\mathbf{y}\in
row\left( \mathbf{A}\right) $ if and only if
\begin{equation}
Minors_{r+1}\left(
\begin{array}{c}
\mathbf{y} \\
\mathbf{A}%
\end{array}%
\right) \mathbf{=0,}  \label{MaxMinEq}
\end{equation}%
where (\ref{MaxMinEq}) means that all minors of the matrix $\left(
\begin{array}{c}
\mathbf{y} \\
\mathbf{A}%
\end{array}%
\right) $ of size $r+1$ are equal to zero.
\end{lemma}

\begin{proof}
As the row vectors of $\mathbf{A}$ are linearly independent, by \cite[%
Corollary 2.8 ]{Fan2014matrix} there is an invertible matrix $\mathbf{P}\in
R^{n\times n}$ such that $\mathbf{AP=}\left(
\begin{array}{cc}
\mathbf{I}_{r} & \mathbf{0}%
\end{array}%
\right) $. Set $\mathbf{P}=\left(
\begin{array}{cc}
\mathbf{P}_{1} & \mathbf{P}_{2}%
\end{array}%
\right) $ where $\mathbf{P}_{1}$ and $\mathbf{P}_{2}$ are submatrices of $%
\mathbf{P}$ of sizes $n\times r$ and $n \times ( n-r)$,
respectively. Assume that (\ref{MaxMinEq}) holds. Then, using the
Cauchy--Binet formula, we get
\begin{equation*}
Minors_{r+1}\left( \left(
\begin{array}{c}
\mathbf{y} \\
\mathbf{A}%
\end{array}%
\right) \mathbf{P}\right) \mathbf{=0},
\end{equation*}%
that is to say,
\begin{equation*}
Minors_{r+1}\left(
\begin{array}{cc}
\mathbf{yP}_{1} & \mathbf{yP}_{2} \\
\mathbf{I}_{r} & \mathbf{0}%
\end{array}%
\right) =\mathbf{0}.
\end{equation*}%
Consequently, using the Laplace expansion along the first row of each $%
\left( r+1\right) \times \left( r+1\right) $ submatrix, we obtain $\mathbf{yP%
}_{2}=\mathbf{0}$. Thus, by Lemma \ref{ParityCheckMat}, $\mathbf{y}\in
row\left( \mathbf{A}\right) $. Conversely, if $\mathbf{y}\in row\left( \mathbf{A}\right) $ then (\ref%
{MaxMinEq}) holds, since $\mathbf{y}$ is a linear combination of the rows of
$\mathbf{A}$.
\end{proof}

Let $\mathbf{A}$ and $\mathbf{y}$ as in Lemma \ref{MaxMinLem}. For any
sequence of $r$ positive integers $1\leq j_{1}<\cdots <j_{r}\leq n$, let $%
a_{j_{1},\ldots ,j_{r}}$ be the determinant of the $r\times r$ submatrix of $%
\mathbf{A}$ with column index in\ $\left\{ j_{1},\ldots ,j_{r}\right\} .$
The set $\left\{ a_{j_{1},\ldots ,j_{r}}:1\leq j_{1}<\cdots <j_{r}\leq
n\right\} $ is said to be a  \textbf{Pl\"ucker coordinates}  \cite%
{Bruns2006determinantal} \ of the free $R-$module $row\left( \mathbf{A}%
\right) $. By \cite[Remark 2.12]{Gorla2017algebraic}, if $\ \mathbf{B}\in
R^{r\times n}$ and $row\left( \mathbf{A}\right) =row\left( \mathbf{B}\right)
$, then there is an invertible matrix $\mathbf{Q}\in R^{r\times r}$ such
that $\mathbf{B}=\mathbf{QA}$. Thus, as in the case of fields, the $R-$module $row\left( \mathbf{A}\right)$ may admit several sets of Pl\"ucker
coordinates, but they are all equal up to a unit multiplicative factor. Moreover, if $R$ is a finite chain
ring, then according to Lemma \ref{DecomFCR}, at least one  component
in any Pl\"ucker coordinates is a unit. Furthermore, by setting $\mathbf{y}=\left(
y_{j_{\alpha }}\right) _{1\leq \alpha \leq n}$ where $y_{j_{\alpha }}\in R$,
and using the Laplace expansion along the first row, Equation (\ref%
{MaxMinEq}) is equivalent to
\begin{equation}
\sum_{\alpha =1}^{r+1}\left( -1\right) ^{\alpha +1}y_{j_{\alpha
}}a_{j_{1},\ldots ,j_{\alpha -1},j_{\alpha +1},\ldots j_{r+1}}=0,
\label{PluckerCoordinates}
\end{equation}%
for all sequence of $r+1$ positive integers $1\leq j_{1}<\cdots <j_{r+1}\leq
n$.

Notice that, when the row vectors of $\mathbf{A}$ are not linearly independent, the ``only if'' part of Lemma \ref{MaxMinLem} may not be true. Indeed,
consider the matrix $\mathbf{A=}\left(
\begin{array}{cc}
2 & 0%
\end{array}%
\right) $ over $%
\mathbb{Z}
_{4}$. Then
\begin{equation*}
Minors_{2}\left(
\begin{array}{cc}
0 & 2 \\
2 & 0%
\end{array}%
\right) =0.
\end{equation*}
But $\left( 0,2\right) \notin row\left( \mathbf{A}\right) $.

Similar to the Support-Minors modeling given in \cite{BBBGNRT20}, we have the
following:

\begin{theorem}
\label{SupportMinorsModeling}Let $\mathbf{M}_{0}$, $\mathbf{M}_{1},\ldots ,%
\mathbf{M}_{k}$ in $R^{m\times n}$, $x_{1},\ldots ,x_{k}$ in $R$ and $r$ in $%
{{\mathbb{N}}}^{\ast }$. 

Set $\mathbf{M}_{x}=\mathbf{M}_{0}+%
\sum_{l=1}^{k}x_{l}\mathbf{M}_{l}$. Then, the following statements are
equivalent.
\begin{itemize} 

\item[(i)] $rk(\mathbf{M}_{x})\leq r$.

\item[(ii)] There exists a Pl\"ucker coordinates $\left\{ z_{j_{1},\ldots ,j_{r}}:1\leq
j_{1}<\cdots <j_{r}\leq n\right\} $ of a free submodule of $R^{n}$ of rank $%
r $ such that
\begin{equation}
\sum_{\alpha =1}^{r+1}\left( -1\right) ^{\alpha +1}\mathbf{M}%
_{x}[i,j_{\alpha }]z_{j_{1},\ldots ,j_{\alpha -1},j_{\alpha +1},\ldots
j_{r+1}}=0,  \label{SupportMinorsModelingEq}
\end{equation}%
for all $i=1,\ldots ,n$ and all sequence of $r+1$ positive integers $1\leq
j_{1}<\cdots <j_{r+1}\leq n$, where $\mathbf{M}_{x}[i,j_{\alpha }]$ is the
entry at the $i^{th}$ row and $j_{\alpha }^{th}$ column of  $\mathbf{M}_x$.
\end{itemize}
\end{theorem}

\begin{proof}
Assume that $rk(\mathbf{M}_{x})\leq r$. Then, by Lemma \ref{Envelope}, there
exists a free submodule $F$ of $R^{n}$ of rank $r$ such that $row\left(
\mathbf{M}_{x}\right) \subset F$. Let $\left\{ z_{j_{1},\ldots ,j_{r}}:1\leq
j_{1}<\cdots <j_{r}\leq n\right\} $ be a Pl\"ucker coordinates of $F$. Then,
by Lemma \ref{MaxMinLem} and (\ref{PluckerCoordinates}), we get (\ref%
{SupportMinorsModelingEq}).

Conversely, assume that (ii) holds. 
Then, by Lemma \ref{MaxMinLem}, all row
vectors of $\mathbf{M}_{x}$ are in a free module of rank $r$. Therefore, by
\cite[Proposition 3.2]{Kamche2019rank}, $rk(\mathbf{M}_{x})\leq r$.
\end{proof}

As stated in Remark \ref{NbEnvelope}, the free submodule $F$ is generally not
unique. Consequently, there are usually several Pl\"ucker coordinates associated to different
free submodules, and which all satisfy Equation (\ref{SupportMinorsModelingEq}).\
Equation (\ref{SupportMinorsModelingEq}) is a system of polynomial equations
with unknowns $x_{l}$ and $z_{j_{1},\ldots ,j_{r}}$. Thus, as specified in
the previous sections, we can use Gr\"obner bases to solve (\ref%
{SupportMinorsModelingEq}). But in some cases, it is possible to use linear
algebra as in \cite{BBBGNRT20}.

\begin{example}
\label{ExSupportMinorsModeling} Consider the MinRank problem (\ref%
{Ex1KSModeling1}) of Example \ref{Ex1KSModeling}. Since $r=1$, then  by Theorem %
\ref{SupportMinorsModeling}, there exists a Pl\"ucker coordinates $\left( z_{1},z_{2},z_{3},z_{4}\right) $  of a free submodule of $\mathbb{Z}_{8}^4$ of rank $1$ such that  (\ref{Ex1KSModeling1}) is equivalent
to
\begin{equation}
\left\{
\begin{array}{c}
\mathbf{M}_{x}[i,1]z_{2}-\mathbf{M}_{x}[i,2]z_{1}=0 \\
\mathbf{M}_{x}[i,1]z_{3}-\mathbf{M}_{x}[i,3]z_{1}=0 \\
\mathbf{M}_{x}[i,1]z_{4}-\mathbf{M}_{x}[i,4]z_{1}=0 \\
\mathbf{M}_{x}[i,2]z_{3}-\mathbf{M}_{x}[i,3]z_{2}=0 \\
\mathbf{M}_{x}[i,2]z_{4}-\mathbf{M}_{x}[i,4]z_{2}=0 \\
\mathbf{M}_{x}[i,3]z_{4}-\mathbf{M}_{x}[i,4]z_{3}=0%
\end{array}%
\right. ,\ \ \ i=1,...,4  \label{ExSupportMinorsModeling1}
\end{equation}%
where $\mathbf{M}_{x}=x_{1}\mathbf{M}_{1}+x_{2}\mathbf{M}_{2}+x_{3}\mathbf{M}%
_{3}$. Since $\mathbb{Z}_{8}$ is a finite chain ring, at least one component of the Pl\"ucker
coordinates $\left( z_{1},z_{2},z_{3},z_{4}\right) $ is a unit. Without loss
of generality, assume that $z_{4}$ is a unit, then in other to recover $x_{1}$, $%
x_{2}$ and $x_{3},$ we rewrite (\ref{ExSupportMinorsModeling1}) as
\begin{equation}
\mathbf{AX=0}  \label{ExSupportMinorsModeling2}
\end{equation}%
where
$\mathbf{\allowbreak X}^{\top }=\left(
\begin{array}{cccccccccccc}
x_{1}z_{1} & x_{2}z_{1} & x_{3}z_{1} & x_{1}z_{2} & x_{2}z_{2} & x_{3}z_{2}
& x_{1}z_{3} & x_{2}z_{3} & x_{3}z_{3} & x_{1}z_{4} & x_{2}z_{4} & x_{3}z_{4}%
\end{array}%
\right)$
and $\mathbf{A}$ is a matrix with entries in $\mathbb{Z}_{8}$. Using SageMath \cite{Sagemath2023}, we can compute the
row echelon form $\widetilde{\mathbf{A}}$ of $\mathbf{A}$ and get
{\small 
\begin{equation*}
\widetilde{\mathbf{A}}=\left(
\begin{array}{rrrrrrrrrrrr}
1 & 0 & 0 & 0 & 0 & 0 & 0 & 0 & 0 & 0 & 0 & 2 \\
0 & 1 & 0 & 0 & 0 & 0 & 0 & 0 & 0 & 0 & 0 & 2 \\
0 & 0 & 2 & 0 & 0 & 0 & 0 & 0 & 0 & 0 & 0 & 2 \\
0 & 0 & 0 & 1 & 0 & 0 & 0 & 0 & 0 & 0 & 0 & 2 \\
0 & 0 & 0 & 0 & 1 & 0 & 0 & 0 & 0 & 0 & 0 & 2 \\
0 & 0 & 0 & 0 & 0 & 2 & 0 & 0 & 0 & 0 & 0 & 2 \\
0 & 0 & 0 & 0 & 0 & 0 & 1 & 0 & 0 & 0 & 0 & 2 \\
0 & 0 & 0 & 0 & 0 & 0 & 0 & 1 & 0 & 0 & 0 & 2 \\
0 & 0 & 0 & 0 & 0 & 0 & 0 & 0 & 2 & 0 & 0 & 2 \\
0 & 0 & 0 & 0 & 0 & 0 & 0 & 0 & 0 & 1 & 0 & 2 \\
0 & 0 & 0 & 0 & 0 & 0 & 0 & 0 & 0 & 0 & 1 & 2 \\
0 & 0 & 0 & 0 & 0 & 0 & 0 & 0 & 0 & 0 & 0 & 4%
\end{array}%
\right)
\end{equation*}
}
Therefore, (\ref{ExSupportMinorsModeling2}) is equivalent to
\begin{equation*}
\widetilde{\mathbf{A}}\mathbf{X=0}
\end{equation*}%
Thus, $x_{1}z_{4}+2x_{3}z_{4}=$ $x_{2}z_{4}+2x_{3}z_{4}=4x_{3}z_{4}=0$.
Since we assumed that $z_{4}$ is a unit, we get $\left(
x_{1},x_{2},x_{3}\right) \in \left\{ \left( 0,0,0\right) ,\left(
4,4,2\right) ,\left( 0,0,4\right) ,\left( 4,4,6\right) \right\} $.
\end{example}

In the case of fields, some conditions have been given in \cite{BBBGNRT20}
to solve (\ref{SupportMinorsModelingEq}) using linear algebra. It will be
interesting to study if these conditions can be extended to rings.

It is important to note that, according to \cite[Proposition 3.4]{Kamche2019rank}, the rank of a matrix
and its transpose are equal. Therefore, the MinRank problem defined with $%
\mathbf{M}_{0}$, $\mathbf{M}_{1}$, $\ldots $, $\mathbf{M}_{k}$ shares the same
solution set with the one defined with $\mathbf{M}_{0}^{\top }$, $\mathbf{M}%
_{1}^{\top }$, $\ldots $, $\mathbf{M}_{k}^{\top }$. Thus, in order to reduce
the number of variables in the algebraic modeling, one can transpose the matrices before solving the MinRank problem, as stated for example in \cite%
{Bardet2022improvement}.

\section{Rank Decoding Problems Over Finite Principal Ideal Rings \label{RDproblem}}

In this section, we will study the algebraic approach for solving the rank
decoding problem over finite principal ideal rings. Over fields, the rank
decoding problem has several algebraic modeling. As specified in \cite[%
Section 4]{Kalachi2023OnTheRank}, the Ourivski-Johansson modeling \cite%
{Ourivski2002new} and the MaxMinors modeling \cite{BBCGPSTV20}\ cannot
extend directly to rings due to zero divisors. We will show that the
Support-Minors modeling \cite{Bardet2022revisiting} and the modeling using
linearized polynomials \cite{Gaborit2016complexity} can be extended to
finite principal ideal rings.

\subsection{Rank Decoding Problem}

To define the rank decoding problem, we must first recall the construction
of a Galois extension of a finite principal ideal ring $R$. As we
specified in Section \ref{LocalRings},\ $R$ can be decomposed into a direct
sum of local rings. Thus, in the following, we assume that $%
R=R_{(1)}\times \cdots \times R_{(\rho )}$ where each $R_{(j)}$ is a finite chain
ring with maximal ideal $\mathfrak{m}_{(j)}$ and residue field $\mathbb{F}%
_{q_{(j)}}$, for $j=1,\ldots ,\rho $. Let $m$ be a non-zero positive integer
and $h_{(j)}\in $ $R_{(j)}\left[ X\right] $ a monic polynomial of degree $m$
such that its projection onto $\mathbb{F}_{q_{(j)}}\left[ X\right] $ is
irreducible. Set $S_{(j)}=R_{(j)}\left[ X\right] /\left( h_{(j)}\right) $,
then by \cite{Mcdonald1974finite} $S_{(j)}$ is a Galois extension of $%
R_{(j)} $ of degree $m$ with Galois group that is cyclic of order $m$. Moreover, $%
S_{(j)}$ is also a finite chain ring with maximal ideal $\mathfrak{M}_{(j)}=%
\mathfrak{m}_{(j)}S_{(j)}$ and residue field $\mathbb{F}_{q_{(j)}^{m}}$. Let us
denote by $\sigma _{(j)}$ a generator of the Galois group of $S_{(j)}$, $%
\sigma =\left( \sigma _{(j)}\right) _{1\leq j\leq \rho }$ and $%
S=S_{(1)}\times \cdots \times S_{(\rho )}$. Then, as specified in \cite%
{Kamche2019rank}, $S$ is a Galois extension of $R$ of degree $m$ with Galois group generated by $\sigma $. Moreover, there exists $h\in R\left[ X%
\right] $ such that $S\cong R\left[ X\right] /\left( h\right) $. An example
of construction of a Galois extension of $%
\mathbb{Z}
_{40}$ of degree $4$ was given in \cite[Example 2.2]{Kalachi2023OnTheRank}.
The following example shows how one can construct a generator of the Galois
group in practice using the Hensel lifting of a primitive polynomial.

\begin{example}
\label{ExGaloisExtension}Let us construct
a degree $3$ Galois extension of $R = \mathbb{Z}_{8}$, and its Galois group. The residue
field of $R$ is $\mathbb{F}_{q} =\mathbb{F}_{2}
$ and the polynomial $g=X^{3}+X+1$ is a primitive polynomial in $\mathbb{F}_{q}%
\left[ X\right] $. Using the Hensel's lemma, we can construct the polynomial $%
h=X^{3}+6X^{2}+5X+7\ \ \in R\left[ X\right] $, such that $\overline{h}=g$
and $h$ divides $X^{q^{m}-1}-1$. Therefore,$\ S=R\left[ X\right] /\left(
h\right) =R\left[ \alpha \right] $ is a Galois extension of $R$ of degree $%
m=3$, where $\alpha =X+\left( h\right) $. Moreover, $\alpha ^{q^{m}-1}=1$
and $\alpha ^{i}\neq 1$, for $0<i<q^{m}-1$. Thus, the Galois group is
generated by the map $\sigma :S\rightarrow S$ given by $\alpha \mapsto
\alpha ^{q}$, that is to say, for all \ $x=\sum_{i=0}^{m-1}x_{i}\alpha ^{i}$, where
$x_{i}\in R$, $\sigma \left( x\right) =\sum_{i=0}^{m-1}x_{i}\alpha ^{iq}$.
\end{example}

\begin{definition}
Let $\ \mathbf{u}=\left( u_{1},\ldots ,u_{n}\right) \in S^{n}$.
\begin{enumerate}
    \item[a)] The \textbf{support} of $\mathbf{u}$, denoted by $supp(\mathbf{u})$, is
the $R-$submodule of $S$ generated by $\{ u_{1},\ldots ,u_{n} \} $.

\item[b)] The \textbf{rank} of $\mathbf{u},$ denoted by $rk_{R}\left( \mathbf{u}%
\right),$ or simply by $rk\left( \mathbf{u}\right) $ is the smallest
number of elements in $supp(\mathbf{u})$ which generate $supp(\mathbf{u})$
as a $R-$module.
\end{enumerate}
\end{definition}

Since $S$ is a free $R-$module, computing the rank of a vector $%
\mathbf{u}\in S^{n}$ can be done by using its matrix representation in a $R-$basis
of $S$ as in the case of finite fields (for more details see \cite[Proposition 3.13]%
{Kamche2019rank}).

\begin{definition}
Let $\mathcal{C}$ be a $S-$submodule of $S^{n}$, $\mathbf{y}$ an element of
$S^{n}$ and $r\in {{\mathbb{N}}}^{\ast }$. The \textbf{rank decoding problem}
is to find $\mathbf{e}$ in $S^{n}$ and $\mathbf{c}$ in $\mathcal{C}$ such
that $\mathbf{y}=\mathbf{c}+\mathbf{e}$ with $rk(\mathbf{e})\leq r$.
\end{definition}

Using the representation of elements in $S^{n}$ as elements of $R^{m\times
n} $, the rank decoding problem can be reduced to the MinRank problem, as in
the case of finite fields \cite{Faugere2008cryptanalysis}.

\begin{example}
\label{ExDecodingProblem}Let us consider the rings $R=\mathbb{Z}_{8}$ and $\ S=R\left[
\alpha \right] $ as in Example \ref{ExGaloisExtension}. Let $\mathcal{C}%
\subset S^{3}$ be the $S-$linear code generated by:
\begin{equation*}
\mathbf{g}=\left( 1,2\alpha ^{2}+\alpha +2,\alpha ^{2}+3\alpha \right) .
\end{equation*}%
Set
$\mathbf{y=}\left( 4\alpha ^{2}+3\alpha +3,\,5\alpha ^{2}+7\alpha
+6,\,2\alpha ^{2}+4\alpha +5\right) 
$ and consider the instance of the rank decoding problem consisting of finding $\mathbf{c}\in \mathcal{C}
$ such that
\begin{equation}
rk\left( \mathbf{y}-\mathbf{c}\right) \leq 1.  \label{ExDecodingProblem1}
\end{equation}%
Equation (\ref{ExDecodingProblem1}) is equivalent to finding $x_{1}$, $x_{2}$%
, $x_{3}$ in $R$ such
\begin{equation}
rk\left( \mathbf{y}-\left( x_{1}+x_{2}\alpha +x_{3}\alpha ^{2}\right)
\mathbf{g}\right) \leq 1  \label{ExDecodingProblem2}
\end{equation}%
Since $\mathcal{C}$ is generated by $\mathbf{g}$, then the matrix
representation of $\mathcal{C}$ in the basis $\left( 1,\alpha ,\alpha
^{2}\right) $ is the $R-$linear code generated by $\mathbf{M}_{1}$, $\mathbf{%
M}_{2}$, $\mathbf{M}_{3}$ which are respectively the representation matrices
of $\mathbf{g}$, $\alpha \mathbf{g}$, $\alpha ^{2}\mathbf{g}$ in the basis $%
\left( 1,\alpha ,\alpha ^{2}\right) $. Let $\mathbf{M}_{0}$ be the matrix
representation of $-\mathbf{y}$ in the basis $\left( 1,\alpha ,\alpha
^{2}\right) $. Then, the rank decoding problem (\ref{ExDecodingProblem2}) is
equivalent to the MinRank problem (\ref{Ex2KSModeling1}) defined in Example %
\ref{Ex2KSModeling}. The solution of (\ref{Ex2KSModeling1}) is $x_{1}=1$, $%
x_{2}=3$, $x_{3}=6$. Thus, $\mathbf{c=}\left( 1+3\alpha +6\alpha ^{2}\right)
$ $\mathbf{g}$.
\end{example}

\subsection{Support-Minors Modeling}

According to \cite[Proposition 3.14]{Kamche2019rank}\ we have the following:

\begin{lemma}
\label{VectorRankDecompositions} For any $\mathbf{u}\in S^{n}$ with $rk\left(
\mathbf{u}\right) \leq r$, there exists $\mathbf{b} \in S^{r}$ and $\mathbf{%
A} \in R^{r\times n}$ such that $row\left( \mathbf{A}\right) $ is a free
module of rank $r$ and $\mathbf{u}=\mathbf{bA}$.
\end{lemma}

The following result is a generalization of the Support-Minors modeling for
the rank decoding problem giving in  \cite{Bardet2022revisiting}.


\begin{theorem}
\label{SupportMinorsModelingRDP} Let $\mathcal{C}$ be a $S-$submodule of $%
S^{n}$ with a generator matrix $\mathbf{G}=\left( g_{i,j}\right) _{1\leq
i\leq k,1\leq j\leq n}$, $\mathbf{y=}\left(
y_{i}\right) _{1\leq i\leq n}\in $ $S^{n}$ and $r\in {{\mathbb{N}}}$. The following statements
are equivalent.

\begin{enumerate}
\item[(i)] There exists $\mathbf{c}\in \mathcal{C}$ such that $rk\left(
\mathbf{y}-\mathbf{c}\right) \leq r$.

\item[(ii)] There exist a set  $\left\{ z_{j_{1},\ldots
,j_{r}}:1\leq j_{1}<\cdots <j_{r}\leq n\right\} $ of Pl\"ucker coordinates of a free submodule of $%
R^{n}$ of rank $r$ and $\mathbf{x}=\left( x_{i}\right) _{1\leq i\leq k}\in
S^{k}$ such that
\begin{equation}
\sum_{s=1}^{r+1}\sum_{i=1}^{k}\left( -1\right) ^{s+1}\left(
x_{i}g_{i,_{j_{s}}}-y_{j_{s}}\right) z_{j_{1},\ldots ,j_{s-1},j_{s+1},\ldots
j_{r+1}}=0,  \label{EqSupportMinorsModelingRDP}
\end{equation}

for all sequence of $r+1$ positive integers $1\leq j_{1}<\cdots <j_{r+1}\leq
n.$
\end{enumerate}
\end{theorem}

\begin{proof}
Using Lemmas\ \ref{MaxMinLem} and \ref{VectorRankDecompositions}, the
proof is similar to the one from \cite{Bardet2022revisiting}.
\end{proof}
\\

Equation (\ref{EqSupportMinorsModelingRDP}) is a system of algebraic
equations over $S$ with unknows $x_{i}\in S$ and $z_{j_{1},\ldots
,j_{s-1},j_{s+1},\ldots j_{r+1}}\in R$. To solve this equation using Gr\"obner
bases, we must first expand this equation to $R$.

\begin{example}
\label{ExSupportMinorsModelingRDP}Consider the rank decoding problem (\ref%
{ExDecodingProblem1}) of Example \ref{ExDecodingProblem}. Set $\mathbf{g=}%
\left( g_{1,}g_{2},g_{3}\right) $ and $\mathbf{y=}\left(
y_{1,}y_{2},y_{3}\right) $. Then, by Theorem \ref{SupportMinorsModelingRDP},
there are $x$ in $S$ and $z_{1}$, $z_{2}$, $z_{3}$ in $R$ such that
\begin{equation}
\left\{
\begin{array}{c}
\left( xg_{1}-y_{1}\right) z_{2}-\left( xg_{2}-y_{2}\right) z_{1}=0 \\
\left( xg_{1}-y_{1}\right) z_{3}-\left( xg_{3}-y_{3}\right) z_{1}=0 \\
\left( xg_{2}-y_{2}\right) z_{3}-\left( xg_{3}-y_{3}\right) z_{2}=0%
\end{array}%
\right.  \label{ExSupportMinorsModelingRDP1}
\end{equation}

Since $R=%
\mathbb{Z}
_{8}$ is a finite chain ring, then at least one of the elements in the Pl\"ucker coordinates $\left( z_{1},z_{2},z_{3}\right) $ is a unit. Without loss
of generality, assume that $z_{1}=1$. Set $x=x_{0}+x_{1}\alpha +x_{2}\alpha
^{2}$ where $x_{i} \in R$, for $i \in \{ 0, 1, 2 \}$. Using SageMath \cite{Sagemath2023}, we
substitute $x$ and \ $z_{1}$ in (\ref{ExSupportMinorsModelingRDP1}) and
expand the resulting equations over $R$ using the basis $\left( 1,\alpha
,\alpha ^{2}\right) $, then we obtain a system of equations of the form%
\begin{equation}
\left\{
\begin{array}{l}
-z_{2}x_{0}+3z_{2}+2x_{0}+2x_{1}+5x_{2}+2=0 \\
-z_{2}x_{1}+3z_{2}+x_{0}+x_{2}+1=0 \\
-z_{2}x_{2}+4z_{2}+2x_{0}+5x_{1}+2x_{2}+3=0 \\
-z_{3}x_{0}+3z_{3}+x_{1}+5x_{2}+3=0 \\
-z_{3}x_{1}+3z_{3}+3x_{0}+3x_{1}+4=0 \\
-z_{3}x_{2}+4z_{3}+x_{0}+5x_{1}+5x_{2}+6=0 \\
z_{2}x_{1}+5z_{2}x_{2}+3z_{2}+6z_{3}x_{0}+6z_{3}x_{1}+3z_{3}x_{2}+6z_{3}=0
\\
3z_{2}x_{0}+3z_{2}x_{1}+4z_{2}-z_{3}x_{0}-z_{3}x_{2}-z_{3}=0 \\
z_{2}x_{0}+5z_{2}x_{1}+5z_{2}x_{2}+6z_{2}+6z_{3}x_{0}+3z_{3}x_{1}+6z_{3}x_{2}+5z_{3}=0%
\end{array}%
\right.  \label{ExSupportMinorsModelingRDP2}
\end{equation}

As in Example \ref{Ex2KSModeling}, when we compute a Gr\"obner basis
associated to (\ref{ExSupportMinorsModelingRDP2}) with the lexicographic
order $z_{2}>z_{3}$ $>x_{0}>x_{1}>x_{2}$ we obtain a system which is not
easy to solve. So, to simplify the resolution, we add the polynomial
expressions $F_{m}\left( z_{2}\right) $, $F_{m}\left( z_{3}\right) $, $%
F_{m}\left( x_{0}\right) $, $F_{m}\left( x_{1}\right) $, $F_{m}\left(
x_{2}\right) $ as in Example \ref{ExSolveFCR_GB2}, and get the Gr\"obner
basis: $z_{2}^{4}-z_{2}^{2}$, $2z_{2}$, $z_{3}^{4}+3z_{3}^{2}+4$, $2z_{3}+4$%
, $x_{0}+7$, $x_{1}+5$, $x_{2}+2$. Thus, $x_{0}=1$, $x_{1}=3$, $x_{2}=6$.
\end{example}

\subsection{Algebraic Modeling With Skew Polynomials}

Skew polynomials \cite{Ore1933theory} generalize linearized polynomials, and
some properties of linearized polynomials have been extended to skew
polynomials in \cite{Kamche2019rank}.

\begin{definition}
The\textbf{\ skew polynomial ring} over $S$ with automorphism $\sigma $,
denoted by $S[X,\sigma ]$, is the ring of all polynomials in $S[X]$ such that
\begin{itemize}
    \item the addition is defined to be the usual addition of polynomials;
    \item the multiplication is defined by the basic rule $Xa=\sigma \left( a\right)
X$, for all $a\in S$.
\end{itemize}

\end{definition}

\begin{notation}[Evaluation Map]
Let $f=a_{0}+a_{1}X+\cdots a_{k}X^{k}\in S[X,\sigma ]$, $x\in S$ and \ $%
\mathbf{u}=\left( u_{i}\right) _{1\leq i\leq n}\in S^{n}$.

\begin{enumerate}
\item $f\left( x\right) :=a_{0}x+a_{1}\sigma \left( x\right) +\cdots
+a_{k}\sigma ^{k}\left( x\right) $.

\item $f\left( \mathbf{u}\right) :=\left( f\left( u_{i}\right) \right)
_{1\leq i\leq n}$.
\end{enumerate}
\end{notation}

According to \cite[Propositions 3.15, 3.16 and Corollary 2.7]%
{Kamche2019rank}, we have the following proposition.

\begin{proposition}
\label{Annulator} For all $\mathbf{u}\in S^{n}$, $rk\left( \mathbf{u}\right)
\leq r$ if and only if there is a monic skew polynomial $f\in S[X,\sigma ]$
of degree $r$ such that $f\left( \mathbf{u}\right) =\mathbf{0}.$ Moreover,
if \ $supp(\mathbf{u})$ is a free module and $rk\left( \mathbf{u}\right) =r$
then $f$ is unique.
\end{proposition}

\begin{remark}
\label{NbAnnulator} To construct the skew polynomial $f$ of Proposition \ref%
{Annulator}, one generally use a free $R-$submodule of $S$ which contains $%
supp(\mathbf{u})$. Hence, as we pointed out in Remark \ref{NbEnvelope} , there are generally more than one
free $R-$submodule of $S$ which contains $supp(\mathbf{u})$. Thus, $f$ is generally not unique.
\end{remark}

\begin{example}
Consider again $R=\mathbb{Z}_{8}$ and $\ S=R\left[ \alpha \right] $ as in
Example \ref{ExGaloisExtension}. The rank of $\mathbf{u}=\left( 2+6\alpha
^{2},0,4+4\alpha ^{2}\right) $ is $1$ and we would like to find all the
monic skew polynomials $f\in S[X,\sigma ]$ of degree $1$ such that $f\left(
\mathbf{u}\right) =\mathbf{0}$. So, $f=X+w$, where $w\in S$ and can be written as $w=w_{0}+w_{1}\alpha +w_{2}\alpha ^{2}$ with $w_{0}$,
$w_{1}$, $w_{2}$ in $R$.
When we solve the equation $f\left( \mathbf{u}\right) =\mathbf{0}$, we get $w_{0} \in \{ 3,7\}$, $w_{1} \in \{0,4 \}$, $w_{2} \in \{3,7 \}$. Thus,
there are eight monic skew polynomials $f\in S[X,\sigma ]$ with degree $1$ such
that $f\left( \mathbf{u}\right) =\mathbf{0}.$
\end{example}

\begin{notation}
If $\mathbf{B}=\left( b_{i,j}\right) $ is a matrix with entries in $S$ and $l
$ is a positive integer then,
\begin{equation*}
\sigma ^{l}\left( \mathbf{B}\right) :=\left( \sigma ^{l}\left(
b_{i,j}\right) \right) .
\end{equation*}
\end{notation}

The following result is a generalization of the result giving in \cite[%
Section V]{Gaborit2016complexity}.

\begin{theorem}
\label{KeyEquationTheorem} Let $\mathcal{C}$ be a $S-$submodule of $S^{n}$
with generator matrix $\mathbf{G}=\left( g_{i,j}\right) _{1\leq i\leq
k,1\leq j\leq n}$, $\mathbf{y=}\left(
y_{i}\right) _{1\leq i\leq n}\in $ $S^{n}$ and $r\in {{\mathbb{N}}}$. The following statements are
equivalent.

\begin{enumerate}
\item[(i)] There exists $\mathbf{c}\in \mathcal{C}$ such that $rk\left(
\mathbf{y}-\mathbf{c}\right) \leq r$.

\item[(ii)] There are $\left( z_{l}\right) _{0\leq l\leq r}\in S^{r+1}$, $%
z_{r}=1$, and $\mathbf{x}=\left( x_{i}\right) _{1\leq i\leq k}\in S^{k}$
such that
\begin{equation}
\sum_{l=0}^{r}z_{l}\sigma ^{l}\left( \mathbf{y}\right)
=\sum_{l=0}^{r}z_{l}\sigma ^{l}\left( \mathbf{xG}\right)  \label{KeyEquation}
\end{equation}
\end{enumerate}

Moreover, if $\mathcal{C}$ is a free $S-$submodule\ of rank $k$ and $r\leq
t$, where $t$ is the error correction capability of $\mathcal{C}$ then $%
\mathbf{x}$ is unique.
\end{theorem}

\begin{proof}
By Proposition \ref{Annulator}, $rk\left( \mathbf{y}-\mathbf{c}\right)  \leq r$
if and only if there exists a monic skew polynomial $P=\sum_{l=0}^{r}z_{l}X^{l}%
\in S[X,\sigma ]$ of degree $r$ such that $P\left( \mathbf{y}-\mathbf{c}%
\right) =\mathbf{0}$. Since \ $\mathbf{c}\in \mathcal{C}$, then there exists $%
\mathbf{x}=\left( x_{i}\right) _{1\leq i\leq k}\in S^{k}$, such that $%
\mathbf{c=xG}$. Thus, the result follows.
\end{proof}
\\

According to Remark \ref{NbAnnulator}, when the support of the error is not
a free module, the unknowns $z_{i}$'s, $i=0,\ldots ,r-1$ are not unique, even
if $\mathbf{x}$ is unique. So in general, \eqref{KeyEquation} has many
solutions. This is the main difference compared to the same result over
finite fields. Note that to solve the rank decoding problem, we don't need
the unknowns $z_{i}$. We just need $\mathbf{x}$, since we can use it to
recover $\mathbf{c}$.

\subsubsection{Solving by Linearization}

In this subsection, we will show that in some cases, the unknowns $\mathbf{x}$
in (\ref{KeyEquation}) can be recovered using linear algebra. Equation %
\eqref{KeyEquation} is equivalent to%
\begin{equation}
\mathbf{Au}=\mathbf{0}  \label{LinearModelingRDP1}
\end{equation}%
where%
\begin{equation*}
\mathbf{A}=\left(
\begin{array}{ccccccc}
-\sigma ^{0}\left( \mathbf{y}^{\top }\right)  & \cdots  & -\sigma
^{r-1}\left( \mathbf{y}^{\top }\right)  & \sigma ^{0}\left( \mathbf{G}^{\top
}\right)  & \cdots  & \sigma ^{r}\left( \mathbf{G}^{\top }\right)  & -\sigma
^{r}\left( \mathbf{y}^{\top }\right)
\end{array}%
\right)
\end{equation*}%
and%
\begin{equation*}
\mathbf{u}^{\top }\mathbf{=}\left(
\begin{array}{ccccccc}
z_{0} & \cdots  & z_{r-1} & z_{0}\sigma ^{0}\left( \mathbf{x}\right)  &
\cdots  & z_{r}\sigma ^{r}\left( \mathbf{x}\right)  & z_{r}%
\end{array}%
\right) .
\end{equation*}

In the same way as the row echelon form over fields, the matrix $\mathbf{A}$
can be decomposed as $\mathbf{A=PT}$ where $\mathbf{P}$ is an invertible
matrix and $\mathbf{T=}\left( t_{i,j}\right) $ is an upper triangular
matrix, that is to say $t_{i,j}=0$ if $i>j$ \cite[Theorem 3.5]%
{Kaplansky1949elementary}.\ The matrix $\mathbf{T}$ is usually called the\
Hermite normal form of $\mathbf{A}$. One can compute the Hermite normal
using the same methods as the Gaussian elimination algorithm, see for
example \cite{Kaplansky1949elementary, Storjohann2000algorithms, Bulyovszky2017polynomial}. As $%
z_{r}=1$, the following proposition shows that if $\mathbf{T}$ has a
specific form, then $\mathbf{x}$ can be recovered.

\begin{proposition}
With the above notations, assume that \eqref{KeyEquation} has a solution and
that $\mathbf{T}$ is of the form%
\begin{equation}
\mathbf{T}=\left(
\begin{array}{cc}
\mathbf{T}_{1} & \mathbf{T}_{2} \\
\mathbf{0} & \mathbf{T}_{3} \\
\mathbf{0} & \mathbf{0}%
\end{array}%
\right)   \label{LinearModelingRDP2}
\end{equation}%
where $\mathbf{T}_{1}$ is an $r(k+1)\times r(k+1)$ upper triangular matrix, $%
\mathbf{T}_{2}$ being a $r(k+1)\times (k+1)$ matrix and $\mathbf{T}%
_{3}=\left(
\begin{array}{cc}
\mathbf{I}_{k} & \mathbf{b}%
\end{array}%
\right) $ where $\mathbf{b}$ is a $k\times 1$ matrix, then
\begin{equation*}
\mathbf{x=}-\sigma ^{-r}\left( \mathbf{b}^{\top }\right).
\end{equation*}
\end{proposition}

Note that (\ref{LinearModelingRDP1}) is a homogeneous system of $n$ linear
equations with $(k+1)(r+1)$ unknowns. So, a necessary condition for $\mathbf{%
T}$ to have the form (\ref{LinearModelingRDP2}) is $n\geq (k+1)(r+1)-1$. The
same condition was given in \cite[Theorem 12]{Gaborit2016complexity} in the
case of finite fields. With this condition, we observed in our simulations
that, when $\mathcal{C}$ is a random free submodule, $\mathbf{x}$ can be
recovered in many cases. It will be therefore interesting to study the
probability of this observation.

\begin{example}
\label{ExLinearRDP}Consider the rank decoding problem of Example \ref%
{ExDecodingProblem}. Then there are $x\in S$ and $\mathbf{e\in }S^{3}$ such
that
\begin{equation}
\mathbf{y}=x\mathbf{g}+\mathbf{e}  \label{ExLinearRDP1}
\end{equation}%
with $rk\left( \mathbf{e}\right) =r=1$. So, the skew polynomial $P\in
S[X,\sigma ]$, such that
\begin{equation}
P\left( \mathbf{e}\right) =\mathbf{0}  \label{ExLinearRDP2}
\end{equation}%
is of the form $P=z_{0}+z_{1}X$ where $z_{0},z_{1}\in S$ with $z_{1}=1$. By setting
$\mathbf{g=}\left( g_{1,}g_{2},g_{3}\right) $ and $\mathbf{y=}\left(
y_{1,}y_{2},y_{3}\right)$,  (\ref{ExLinearRDP1}) and (\ref%
{ExLinearRDP2}) imply%
\begin{equation}
z_{0}\left( xg_{j}-y_{j}\right) +z_{1}\sigma \left( xg_{j}-y_{j}\right) =0,\
\ \ \ j=1,...,3.  \label{ExLinearRDP3}
\end{equation}%
which means that
\begin{equation}
\mathbf{A}\left(
\begin{array}{c}
z_{0} \\
z_{0}x \\
z_{1}\sigma \left( x\right) \\
z_{1}%
\end{array}%
\right) =\left(
\begin{array}{c}
0 \\
0 \\
0 \\
0%
\end{array}%
\right)  \label{ExLinearRDP4}
\end{equation}%
where
\begin{equation*}
\mathbf{A}=\left(
\begin{array}{cccc}
-y_{1} & g_{1} & \sigma \left( g_{1}\right) & -\sigma \left( y_{1}\right) \\
-y_{2} & g_{2} & \sigma \left( g_{2}\right) & -\sigma \left( y_{2}\right) \\
-y_{3} & g_{3} & \sigma \left( g_{3}\right) & -\sigma \left( y_{3}\right)%
\end{array}%
\right) .
\end{equation*}%
Using Magma \cite{Cannon2013handbook}, we compute the row echelon form of $\mathbf{A}$ and get:
\begin{equation*}
\mathbf{T}=\left(
\begin{array}{cccc}
1 & \alpha ^{2}+\alpha & 0 & 2\alpha ^{2}+4 \\
0 & 2 & 0 & 6\alpha ^{2}+4\alpha \\
0 & 0 & 1 & 3\alpha ^{2}+6\alpha +3%
\end{array}%
\right)
\end{equation*}%
Thus,
\begin{eqnarray*}
x &=&-\sigma ^{-1}\left( 3\alpha ^{2}+6\alpha +3\right) \\
&=&1+3\alpha +6\alpha ^{2}
\end{eqnarray*}
\end{example}

\subsubsection{Solving With Gr{\"o}bner Bases \label{SolvingWithGB2}}

When $S$ is a finite field, Equation (\ref{KeyEquation}) is a system of
multivariate polynomial equations in the variables $z_{l}$ and $x_{i}$, and such a system was solved directly with Gr\"obner bases in \cite[Section VII]{Gaborit2016complexity}. However, when $S$ is not a field,
the expression $\sigma ^{l}\left( x_{i}g_{i,j}\right) $ is not a polynomial
function in the variable $x_{i}$. So, to transform (\ref{KeyEquation}) into
a system of multivariate polynomial equations, we will expand this equation
in $R$. Let $\left( \beta _{u}\right) _{1\leq u\leq m}$ be a $R-$basis of $S$. Using the notations of Theorem \ref{KeyEquationTheorem}, set
$x_{i}=\sum_{u=1}^{m}x_{i,u}\beta _{u}$ and $z_{l}=\sum_{v=1}^{m}z_{l,v}%
\beta _{v}$ where $x_{i,u}$ and $z_{l,v}$ are in $R$. If we substitute $x_{i}
$ and $z_{l}$ in (\ref{KeyEquation}) and expand the resulting equations over
$R$ using the basis $\left( \beta _{u}\right) _{1\leq u\leq m}$, then we
obtain a system of equations of the form:%
\begin{equation}
\left( \widetilde{\mathbf{x}}\otimes \widetilde{\mathbf{z}}\right) \mathbf{A}%
+\widetilde{\mathbf{x}}\mathbf{B}+\widetilde{\mathbf{z}}\mathbf{C}+\mathbf{D}%
=\mathbf{0}  \label{Expa}
\end{equation}%
where $\widetilde{\mathbf{x}}=\left( x_{1,1},\ldots x_{1,m},\ldots
,x_{k,1},\ldots x_{k,m}\right) $, $\widetilde{\mathbf{z}}=\left(
z_{0,1},\ldots z_{0,m},\ldots ,z_{r-1,1},\ldots z_{r-1,m}\right) $, and $%
\mathbf{A}$, $\mathbf{B}$, $\mathbf{C}$, $\mathbf{D}$ are matrices with $mn$
columns and entries in $R$.

Assume that $\mathcal{C}$ is a free $S-$submodule and $r\leq t$, where $t $
is the error correction capability of $\mathcal{C}$. Then, according to
Theorem \ref{KeyEquationTheorem}, Equation (\ref{Expa}) has a unique
solution in the variables $\widetilde{\mathbf{x}}$ that we denote by $%
\widetilde{\mathbf{x}}_{0}$. Remember that when the support of the error is
not a free module, Equation (\ref{Expa}) has many solutions in the variables
$\widetilde{\mathbf{z}}$. But also note that we do not need all the
solutions of (\ref{Expa}). We just need the partial solution $\widetilde{%
\mathbf{x}}_{0}$. Therefore, to solve (\ref{Expa}) we can use the
elimination theorem as specified in Subsection \ref{SolvingWithGB1} to
simply find the partial solution $\widetilde{\mathbf{x}}_{0}$ using Gr\"obner bases. Equation (\ref{Expa}) is equivalent to
\begin{equation*}
\left\{
\begin{array}{c}
\left( \widetilde{\mathbf{x}}_{0}\otimes \widetilde{\mathbf{z}}\right)
\mathbf{A+}\widetilde{\mathbf{x}}_{0}\mathbf{B}+\widetilde{\mathbf{z}}%
\mathbf{C+D}=\mathbf{0} \\
\widetilde{\mathbf{x}}-\widetilde{\mathbf{x}}_{0}=\mathbf{0}%
\end{array}%
\right.
\end{equation*}%
In our simulations we observe that, when $\mathcal{C}$ is a random free
submodule, the ideal associated to (\ref{Expa}) is generated by
\begin{equation*}
\left\{ \widetilde{\mathbf{x}}-\widetilde{\mathbf{x}}_{0},\ \left(
\widetilde{\mathbf{x}}_{0}\otimes \widetilde{\mathbf{z}}\right) \mathbf{A+}%
\widetilde{\mathbf{x}}_{0}\mathbf{B}+\widetilde{\mathbf{z}}\mathbf{C+D}%
\right\} .
\end{equation*}%
Thus, we have the partial solution $\widetilde{\mathbf{x}}_{0}$ directly
when we compute a Gr\"obner basis of (\ref{Expa}). It is therefore
interesting to study the properties of (\ref{Expa}) in order to estimate
the probability of this observation.

\begin{example}
\label{ExGrobnerRDP}Consider Equation (\ref{ExLinearRDP3}) of Example \ref%
{ExLinearRDP}. Set $x=x_{0}+x_{1}\alpha +x_{2}\alpha ^{2}$ and $%
z_{0}=t_{0}+t_{1}\alpha +t_{2}\alpha ^{2}$ where $x_{i}$ and \ $t_{i}$ are
in $R$ for $i=1,\ldots ,3$. Using SageMath, we substitute $x$, $z_{0}$ and $%
z_{1}=1$ in (\ref{ExLinearRDP3}) and expand the resulting equations over $R$
using the basis $\left( 1,\alpha ,\alpha ^{2}\right) $ to finally obtain a
system of equations of the form

\begin{equation}
\tiny{
\left\{
\begin{array}{l}
x_{0}t_{0}+x_{2}t_{1}+x_{1}t_{2}+2x_{2}t_{2}+x_{0}+2x_{2}+5t_{0}+4t_{1}+5t_{2}+5=0
\\
x_{1}t_{0}+x_{0}t_{1}+3x_{2}t_{1}+3x_{1}t_{2}-x_{2}t_{2}-x_{2}+5t_{0}+t_{1}+3t_{2}+4=0
\\
x_{2}t_{0}+x_{1}t_{1}+2x_{2}t_{1}+x_{0}t_{2}+2x_{1}t_{2}-x_{2}t_{2}+x_{1}-x_{2}+4t_{0}+5t_{1}+3t_{2}+1=0
\\
2x_{0}t_{0}+2x_{1}t_{0}+5x_{2}t_{0}+2x_{0}t_{1}+5x_{1}t_{1}+2x_{2}t_{1}+5x_{0}t_{2}+2x_{1}t_{2}+5x_{2}t_{2}+6x_{0}+4x_{1}+x_{2}+2t_{0}+3t_{1}-t_{2}=0
\\
x_{0}t_{0}+x_{2}t_{0}+x_{1}t_{1}+3x_{2}t_{1}+x_{0}t_{2}+3x_{1}t_{2}+x_{2}t_{2}+6x_{0}+3x_{1}+6x_{2}+t_{0}+3t_{1}+5=0
\\
2x_{0}t_{0}+5x_{1}t_{0}+2x_{2}t_{0}+5x_{0}t_{1}+2x_{1}t_{1}+5x_{2}t_{1}+2x_{0}t_{2}+5x_{1}t_{2}+5x_{2}t_{2}-x_{0}+3x_{1}-x_{2}+3t_{0}-t_{1}+t_{2}+6=0
\\
x_{1}t_{0}+5x_{2}t_{0}+x_{0}t_{1}+5x_{1}t_{1}+5x_{2}t_{1}+5x_{0}t_{2}+5x_{1}t_{2}+2x_{2}t_{2}+2x_{0}+3x_{1}-x_{2}+3t_{0}+6t_{1}+7=0
\\
3x_{0}t_{0}+3x_{1}t_{0}+3x_{0}t_{1}+4x_{2}t_{1}+4x_{1}t_{2}+3x_{2}t_{2}-x_{0}+3x_{1}+3x_{2}+4t_{0}+5t_{1}+6t_{2}+2=0
\\
x_{0}t_{0}+5x_{1}t_{0}+5x_{2}t_{0}+5x_{0}t_{1}+5x_{1}t_{1}+2x_{2}t_{1}+5x_{0}t_{2}+2x_{1}t_{2}+2x_{0}+6x_{1}+3x_{2}+6t_{0}+5t_{2}+6=0%
\end{array}%
\right.  \label{ExGrobnerRDP1}
}
\end{equation}

 \ \

Using SageMath \cite{Sagemath2023}, we compute a Gr\"obner basis of (\ref{ExGrobnerRDP1}) and
get:
\begin{equation*}
\left\{ x_{0}+7,x_{1}+5,x_{2}+2,2t_{0}+2,2t_{1},2t_{2}+2\right\} \text{.}
\end{equation*}%
Thus, $x=x_{0}+x_{1}\alpha +x_{2}\alpha ^{2}=1+3\alpha +6\alpha ^{2}$.
\end{example}

\section{Conclusion}\label{conclusion}
In this work, we have shown that solving systems of algebraic equations
over finite commutative rings reduces to the same problem over Galois rings. Then, using
the elimination theorem and some properties of canonical generating systems,
we have also shown how Gr\"obner bases can be used to solve systems of algebraic
equations over finite chain rings. As applications, these results have been used to give some algebraic approaches for solving the MinRank problem and the rank decoding problem over finite principal ideal rings.

The above work clearly open the door to an important complexity question, namely the real coast of Gr\"obner bases computation over finite chain rings, or at least the coast when dealing with the MinRank and rank decoding problems over finite chain rings.

Another metric used in coding theory and cryptography is the Lee metric \cite%
{Lee1958some}. This metric is usually defined over integer residue rings,
which are specific cases of finite principal ideal rings. Another interesting perspective will be to study the possibility of using algebraic techniques for solving the decoding problem in the Lee metric.

\bibliographystyle{abbrv}
\bibliography{main}

\end{document}